\documentclass[a4paper,12pt]{article}

\usepackage{amsmath,amssymb,amsthm}
\usepackage[T1]{fontenc}

\usepackage{slashbox}

\usepackage{float}
\usepackage{subfig}
\usepackage[a4paper]{geometry}

\usepackage{array}

\usepackage{stmaryrd}

\usepackage{ifpdf}
\ifpdf
   \usepackage[pdftex,bookmarks=true,bookmarksnumbered,%
    plainpages=false,hypertexnames=false,pdfpagelabels,%
    colorlinks=true,linkcolor=black,urlcolor=blue,citecolor=red,%
    pdfstartview=FitH]{hyperref}
   \usepackage[pdftex]{graphicx,color}
\else
    \usepackage[dvips,bookmarks=true,bookmarksnumbered,plainpages=false,%
     colorlinks=true,linkcolor=black,urlcolor=blue,citecolor=red,%
     hypertexnames=false,pdfstartview=FitH,breaklinks=true]{hyperref}
    \usepackage{pstricks,pst-node,pst-text,pst-3d}
    \usepackage[dvips]{graphicx}
    \usepackage{draftcopy}
\fi

\newtheorem{theorem}{Theorem}
\newtheorem{prop}[theorem]{Proposition}
\newtheorem{lemma}[theorem]{Lemma}
\theoremstyle{definition}
\newtheorem{defn}[theorem]{Definition}
\newtheorem{hypo}[theorem]{Hypothesis}

\theoremstyle{remark}
\newtheorem{rem}[theorem]{Remark}

 \newcommand\overrel[2]{\mathrel{\mathop{#2}\limits^{#1}}}

\newcommand{\bfn}{\boldsymbol{n}}
\newcommand{\balpha}{\boldsymbol{\alpha}}
\newcommand{\bpi}{\boldsymbol{\pi}}

\newcommand{\alphanjji}{\alpha^{\bfn_{(j)}}(j,i)}
\newcommand{\alphanjjik}{\alpha^{\bfn_{(j)}}_\gamma(j,i)}

\newcommand{\piji}{\pi(j,i)}
\newcommand{\piij}{\pi(i,j)}
\newcommand{\pikj}{\pi(k,j)}
\newcommand{\pijN}{\pi(j,N)}

\definecolor{rouge2}{RGB}{230,68,57}  % red S
\definecolor{darkred}{RGB}{150,1,100}  % red S

\usepackage{enumerate}
\usepackage{bbm}
\newcommand{\Ind}{\mathbbm{1}}

\usepackage[round]{natbib}

\usepackage{authblk}
\begin{document}
\title{Liquidity costs: a new numerical methodology and an empirical study\footnote{This work is issued from a collaboration Contract between Inria Tosca team and 
Interests Rates and Hybrid Quantitative Research team of Credit Agricole CIB.}}

\author{Christophe Michel$^1$,
  Victor Reutenauer$^2$,\\
 Denis Talay$^3$ and  Etienne Tanr\'e$^3$
}

\footnotetext[1]{CA-CIB, \texttt{christophe.michel@ca-cib.com}}
\footnotetext[2]{Fotonower (former CA-CIB), \texttt{victor@fotonower.com}}
\footnotetext[3]{Inria, \texttt{Denis.Talay@inria.fr}, \texttt{Etienne.Tanre@inria.fr}}

\maketitle


\begin{abstract}
We consider rate swaps which pay a fixed rate against a floating rate
in presence of bid-ask spread costs.
Even for simple models of bid-ask spread costs,
there is no explicit strategy
optimizing an expected function of the hedging error.
We here propose an efficient algorithm
based on the stochastic gradient method
to compute an approximate optimal strategy without solving a 
stochastic control problem.
We validate our algorithm by numerical experiments.
We also develop several variants of the algorithm
and discuss their performances 
in terms of the numerical parameters and the liquidity cost.
\end{abstract}

	\section{Introduction}
	Classical models in financial mathematics usually assume that 
	markets are perfectly liquid.
	In particular, each trader can buy or
	sell the amount of assets he/she needs at the same price
	(the ``market price''), 
	and the trader's decisions do not affect the price of the asset.
	In practice, the assumption of perfect liquidity is never satisfied
	but the error due to illiquidity is generally negligible with respect to
	other sources of error such as model error or calibration error, etc. 
	
	However, the perfect liquidity assumption does not hold true in practice
		for
		interest rate derivatives market: 
		the liquidity costs to hedge interest rate derivatives 
		are highly time varying. Even though there exist maturities for which
		zero-coupon
		bonds are liquid, bonds at intermediate maturities may be extremely illiquid
		(notice that the underlying interest rate is not directly exchangeable).
	%on the one hand, \eot{the classical assumption of perfect liquidity of the
	%underlying is not fulfilled for an interest rate market. 
	%Indeed, }the underlying interest rate is not directly exchangeable.
	%On the second hand, the liquidity costs to hedge interest rate derivatives 
	%are highly time varying (even though there exist maturities for which
	%zero-coupon
	%bonds are liquid, bonds at intermediate maturities may be extremely illiquid).
	Therefore, hedging such derivatives absolutely needs to take liquidity
	risk into account.
	In this context, defining and computing efficient approximate perfect hedging
	strategies is a complex problem. The main purpose of this paper is to
	show that stochastic optimization methods are powerful tools to treat it
	without solving a necessarily high dimensional stochastic control problem,
	under the constraints that practitioners need to trade at prescribed dates
	and that relevant strategies depend on a finite number of parameters.
	More precisely, we construct and analyze an efficient original numerical 
	method which provides practical strategies facing liquidity costs and minimizing
	hedging errors.

	The outline of the paper is as follows. Section~\ref{sec oursettings} introduces
	the model. 
	In Section~\ref{sec gaussianframework}, we 
	present our numerical method and analyze it from a theoretical point of view
	within the framework of a Gaussian yield curve model.
	Section~\ref{sec numericalwithoutliquidity} is devoted to a numerical
	validation in the idealistic perfect liquidity context. 
	In Section~\ref{sec numericalwithliquidity}, we develop an empirical study 
	of the efficiency of our algorithm in the presence of liquidity costs.

	\section{Our settings: swaps with liquidity cost}\label{sec oursettings}
	\subsection{A short reminder on swaps and swaptions
		%swap options 
		hedging without liquidity cost}\label{sub:rappelsswap}
	One of the most common swaps on the interests rate market is as follows.
	The counterparts  exchange two coupons: the first one is generated
	by a bond (with a constant fixed interest rate) and the second 
	one is generated by a  floating rate (e.g. a LIBOR).
	\begin{defn}\label{def1}
		In a perfectly liquid market, the price at time \(t\) of a zero-coupon bond 
		paying \(1\) at time \(T\)
		is denoted by \(B(t,T)\). 
		The linear forward rate
		\(L(T_F,T_B,T_E)\) is the 
		fair rate decided at time \(T_F\) determining the amount at time \(T_E\)
			obtained by investing 1 at time \(T_B\).
			% price decided at time \(T_F\) and
			% paid at time \(T_B\) to receive \(1\) at time \(T_E\). 
			The following relation is satisfied:
		\begin{equation}\label{eq:lienzerocouponforward}
			L \left( T_F, T_B, T_E \right) = \dfrac{1}{T_E - T_B} \left( \dfrac{B
				\left(T_F, T_B \right) }{B \left(T_F, T_E \right)} - 1 \right).
		\end{equation}
	\end{defn}
	
	A swap contract specifies:
	\begin{itemize}
		\item an agreement date \(t\)
		\item a time line  \((t\leq )~T_0 < \cdots < T_N\)
		\item a fixed interest rate \(r\)
		\item a floating interest rate 
		\item the payoff at each time \(T_i\) (\(1\leq i\leq N \)), that is,
		\begin{equation}\label{eq payoff of the swap}
			P(i) := \left(T_i - T_{i-1} \right) \left( r - L \left( T_{i-1}, T_{i-1}, T_i
			\right)\right).
		\end{equation}
		From \eqref{eq:lienzerocouponforward}, we deduce the equivalent
		expression
		\begin{equation}\label{eq:deuxiemeformepayoffswap}
			P(i) = r\left(T_i - T_{i-1} \right)-\dfrac{1}{B(T_{i-1},T_i)}+1.
		\end{equation}
	\end{itemize}
	In the sequel, we consider that the fixed rate \(r\) is chosen at the money
	(thus the swap 
	at time \(t\) has zero value), and that 
	the swap fixed coupons are received by the trader.

	In the idealistic framework of a market without liquidity cost,
	the trader buys or sells quantities of zero-coupon bonds at 
	the same price (i.e. the \textit{market price}), and there
	exists a discrete time perfect hedging strategy which is
	independent of any model of interest rates.
	In view of \eqref{eq:deuxiemeformepayoffswap}, the replication of the payoff
	\(P(i)\) 
	at time \(T_i\) can be split into three parts:
	\begin{itemize}
		\item the fixed part \(r(T_i-T_{i-1})\) 
		is replicated statically at time \(t\) by selling 
		\(r(T_i-T_{i-1})\) zero-coupon 
		bonds with maturity \(T_i\).
		\item 
		the floating part \(1/B(T_{i-1}, T_i)\) 
		is replicated dynamically at time \(T_{i-1}\) by 
		buying \(1/B(T_{i-1}, T_i)\) zero-coupon bonds with maturity \(T_i\).
		The price of this transaction is equal to \(1\).
		\item the last (fixed) part \(1\) is used at time \(T_i\)
		to buy \(1/B(T_{i}, T_{i+1})\) zero-coupon bonds with maturity \(T_{i+1}\).
	\end{itemize}
	It is easy to see that this strategy is self-financing at times \(T_1, \cdots,
	T_{N-1}\).
	To make it self-financing  at time \(t\) also, at this date one buys  
	\(1\) zero-coupon bond with maturity \(T_0\) and sells
	\(1\) zero-coupon bond with maturity \(T_N\).

	%To summarize, in the idealistic framework, we do not need
	%hedging strategies within the set of all the
	%\((\mathcal{F}^{R}_\theta,\theta \geq 0)\) adapted processes, where 
	%\((\mathcal{F}^{R}_\theta,\theta \geq 0)\) 
	%is the filtration generated by the short rate \((R_\theta,\theta\geq 0)\),
	%and we may restrict the admissible strategies to
	%be adapted to the filtration 
	%generated by the 
	%rates 
	%at times \(\{t, T_0, T_1, \cdots, T_N\}\).
	To summarize, in the idealistic framework, we do not need
		to consider 
		\((\mathcal{F}_\theta,\theta \geq 0)\) adapted hedging strategies, where 
		\((\mathcal{F}_\theta,\theta \geq 0)\) 
		is the filtration generated by the observations (short rates, derivative
		prices, etc.) in continuous time,
		and we may restrict the admissible strategies to
		be adapted to the filtration 
		generated by the 
		observations at times \(\{t, T_0, T_1, \cdots, T_N\}\).

	\subsection{Hypotheses on markets with liquidity
		costs}\label{subsec:hedgingswap}
	We now consider markets with liquidity costs and need to precise 
	our liquidity cost model. In all the sequel \(T_{-1}\) denotes \(t\).

	\begin{hypo}~\\ \label{hypo:controlemesurable}
		We assume that, for all $-1\leq j<i\leq N$,  the number \(
		\piji
		\) of zero-coupon 
		bonds with maturity \(T_i\) bought or
		sold at time \(T_j\) is measurable with respect to
		the filtration generated by \((R_t, R_{T_0},\cdots,R_{T_j})\). That means that 
		the admissible strategies do not depend on the evolution
		of the rate \(R_\theta\) between two tenor dates \(T_m\) and \(T_{m+1}\).
	\end{hypo}

	Denote by \(\Psi(T,U,\pi)\) the buy or sell price for \(\pi\) zero
	coupon bonds. In perfectly liquid markets, \(\Psi(T,U,\pi)\) is the
	linear function \(B(T,U)\pi\), where \(B(T,U)\)
	is defined in Definition~\ref{def1}.
	In the presence of liquidity costs, \(\Psi(T,U,\pi)\) becomes a non-linear
	function of
	\(\pi\).
	\begin{hypo}~\\ \label{hypo:regularitepsi}
		For all \(T\) and \(U\), the price \(\Psi(T,U,\pi)\) is a \(C^1(\mathbb{R})\),
		increasing, convex 
		one-to-one map of \(\pi\) from \(\mathbb{R}\) to \(\mathbb{R}\), and
		\(\Psi(T,U,0)=0\).
	\end{hypo}
	Under the preceding hypothesis, the function \(\Psi\) is positive  when \(\pi >
	0\) and negative 
	when \(\pi < 0\).
	
	In the context of the swap
	%\eot{swaptions},
	%swap options, 
	we set
	\begin{equation}\label{eq psiij}
		\Psi_{i,j}(\pi) := \Psi(T_i,T_j,\pi)
	\end{equation}
	and we only consider self-financing strategies, that is, satisfying
	\begin{equation}\label{eq selffinancing}
		\forall\, 0\leq j\leq N-1, \quad \sum_{-1\leq k<j}
		\pikj
		+
		P(j)
		= \sum_{j<i\leq N}\Psi_{j,i}(
		\piji
		).
	\end{equation}

	\subsection{Optimization objective}
	In the presence of liquidity risk, the market is 
	not complete any more
	%no more 
	%complete 
	and the
	practitioners need to build a strategy which minimizes a given function $S$
	(e.g. a risk measure) of the hedging error.
	Such strategies are usually obtained by solving stochastic control problems.
	These problems require high complexity numerical algorithms
	which are too slow to be used in practice.
	We here propose an efficient and original numerical method to compute 
	approximate optimal strategies.
	As the perfect hedging leads to a null portfolio at time \(T_N\), we have
	to solve the optimization problem
	\begin{equation}\label{eq:optimizationgenerale}
		\inf_{\bpi\in\Pi}
		\mathbb{E}\left[S(W^{\bpi})\right],
	\end{equation}
	where \(W^{\bpi}\) is the terminal wealth (at time \(T_N\)) given the
	strategy \(\bpi\) in the set \(\Pi\) of admissible strategies.

	\section{Hedging error minimization method in a Gaussian framework}\label{sec
		gaussianframework}
	The methodology we introduce in this section is based on the two
	following key observations:
	\begin{enumerate}
		\item[(1)]
		We consider strategies and portfolios with finite second moment, and thus
		optimize
		within \(L^2(\mu)\) for some probability measure \(\mu\).
		The Gram-Schmidt procedure provides countable orthogonal bases \(\mathcal{B}\)
		of the separable Hilbert space \(L^2(\mu)\).
		Our set \(\Pi\) of admissible strategies is obtained by truncating of a given
		basis,
		which reduces the a priori infinite dimensional optimization problem 
		\eqref{eq:optimizationgenerale} to a finite dimensional parametric optimization
		problem of the type
		\(\inf_{\theta\in\Theta} \mathbb{E}\Psi(\theta,X)\),
		where \(\Theta\) is a subset of \(\mathbb{R}^p\), \(X\) is a given random
		variable, \(\Psi\) is a 
		convex function of \(\theta\).
		\item[(2)]
		The Robbins-Monro algorithm and its Chen extension
		are  stochastic alternatives to Newton's method
		to numerically solve such optimization problems.
		These algorithms 
		%avoid 
		do not require
		to compute  \(\tfrac{d}{d\theta}\mathbb{E}\Psi(\theta,X)\).
		They are based on sequences of the type
		\begin{equation}\label{eq:evoltionrobinsmonro}
			\theta_{\gamma+1} = \theta_\gamma - \rho_{\gamma+1}\dfrac{\partial}{\partial
				\theta}\Psi(\theta_\gamma,X_{\gamma+1}),
		\end{equation}
		where \((\rho_\gamma,\gamma\geq 1)\) is  a decreasing sequence and
		\((X_\gamma,\gamma\geq 1)\) is an 
		i.i.d. sequence of random variables distributed as \(X\).
	\end{enumerate}

	We here consider the case of swap
	%swap options  
	in the context of a Gaussian
	yield curve. 
	This assumption is restrictive from a mathematical point of view but is
	satisfied by widely 
	used interest rate models such as Vasicek model, Gaussian affine models
	or HJM (Heath, Jarrow and Morton) models with deterministic volatilities
		(see e.g. \cite{GLT2010} and \cite{musiela_rutkowski_2005}). In
	\citet{gen_vasicek_1999} it is shown 
	that using a three dimensional Gaussian model is sufficient to fit the term
	structure 
	of interest rate products.
	
	\subsection{Step 1: finite dimensional projections of the admissible controls
		space}\label{subsection:finitiedimensionalprojection}
	Consider a Gaussian 
	%interest 
	short
	rate  model \((R_\theta, \theta\geq 0)\)
	: we either suppose that the dynamics of the short rate model is given,
		or that it is deduced from a forward rate model such as in the HJM approach
		for term structures: see e.g. Eq (6.9) in \cite{GLT2010} under the additionnal
		assumption that
		the forward rate volatilities are deterministic.
		In all cases, the resulting bond price model is log-normal.
	
	In view of Hypothesis~\ref{hypo:controlemesurable}, each
	control \(
	\piji
	\) belongs to the Gaussian space
	generated by \((R_t,R_{T_0},\cdots,R_{T_j})\) or,
	equivalently, to a space generated by \(\ell(j)+1\) standard 
	independent Gaussian
	random variables
	\(G^{(0)},G^{(1)},\cdots,G^{(\ell(j))}\) (with 
	\(\ell(j) = j\) for one-factor models, \(\ell(j) = 2j+1\) for two-factor models,
	etc.)  and \(R_t\). 
	An explicit \(L^2\) orthonormal basis 
	of the space generated by  \(G^{(0)},G^{(1)},\cdots,G^{(\ell(j))}\) is
	\begin{equation}\label{eq baseorthonormalehermite}
		\left(\prod_{m=0}^{\ell(j)}\dfrac{H_{n_m}(G^{(m)})}{\sqrt{n_m!}}\right)_{(n_0,\cdots,n_{\ell(j)})
			\in\mathbb{N}^{\ell(j)+1}},
	\end{equation}
	where \((H_n, n\geq 0)\) are the Hermite polynomials
	\[
	H_n(x) = (-1)^ne^{x^2/2}\frac{d^n}{dx^n}e^{-{x^2}/{2}}
	\]
	\citep[see e.g. ][p.236]{Malliavin_IP}.
	
	Thus, the quantities of zero-coupon bonds bought by the trader can be written as
	\begin{equation}\label{decompositionL2}
		\sum_{\bfn_{(j)}\in\mathbb{N}^{\ell(j)+1}}
		\alphanjji
		\prod_{m=0}^{\ell(j)}\dfrac{H_{n_m}(G^{(m)})}{\sqrt{n_m!}},
		\quad 
		\bfn_{(j)} = (n_0,\cdots,n_{\ell(j)}),
	\end{equation}
	where the infinite sum has an \(L^2\) limit sense. 
	A strategy can now be defined as a sequence of real numbers 
	\(\alphanjji\)
	for all \(-1\leq j<i\leq N-1\) and 
	\(\bfn_{(j)} \in\mathbb{N}^{\ell(j)+1}\).

	In order to be in a position to solve a finite dimensional optimization problem,
	we truncate the sequence 
	\((\alphanjji)\).
	Then a strategy is defined by a finite number of real parameters
	\(\{
	\alphanjji, -1\leq j<i\leq N-1\}\) 
	where, for all \(j\), \(\bfn_{(j)}\) belongs to
	a finite subset \(\Lambda^{(j)}\) of \(\mathbb{N}^{\ell(j)+1}\).
	The truncated quantities of zero-coupon bonds bought by the trader write
	\begin{equation}\label{eq:truncation}
		\piji
		= \sum_{\bfn_{(j)}\in\Lambda^{(j)}}
		\alphanjji
		\prod_{m=0}^{\ell(j)}\dfrac{H_{n_m}(G^{(m)})}{\sqrt{n_m!}}.
	\end{equation}

	We discuss the efficiency of this truncation and its convergence 
	in Sections~\ref{subsec35} and \ref{SectionTruncationSpace}.
	
	To simplify, we denote by 
	\(\balpha = \left(
	\alphanjji
	\right)_{i, j, \bfn_{(j)}}\)  
	the parameters to optimize in \(\mathbb{R}^p\) 
	(where the dimension \(p\) is known 
	for each truncation \((\Lambda^{(j)},j=-1,0,\cdots,N-1))\),
	by \(\bpi\left(\balpha \right)\) 
	(or, when no confusion is possible, simply \(\bpi\))
	the hedging strategy
	corresponding to a vector \(\balpha\in\mathbb{R}^p\), see \eqref{eq:truncation}.
	Given the strategy \(\bpi = \bpi\left(\balpha \right)\), the terminal 
	wealth \(W^{(\balpha)}\) (at time \(T_N\)) 
	satisfies
	\begin{equation}
		W^{(\balpha)}  
		= \sum_{-1\leq j<N}
		\pi(\balpha)(j,N)
		+ P(N)
		= \sum_{-1 \leq j<N}
		\pijN
		+ P(N).
		\label{eqn_final_wealth}
	\end{equation}
	The problem \eqref{eq:optimizationgenerale} is now formulated
	as: find \(\balpha^*\) in \(\mathbb{R}^p\) such that
	\begin{equation}
		\label{eq:optimizationparametric}
		\mathbb{E}\left[ S\left(W^{(\balpha^*)}\right)\right] =
		\inf_{\balpha\in\mathbb{R}^p}\mathbb{E}\left[
		S\left(W^{(\balpha)}\right)\right].
	\end{equation}
	\subsection{Step 2: stochastic optimization}
	Using the self-financing equation \eqref{eq selffinancing} one can express
	\(W^{(\balpha)}\)
	as a function of \(\balpha\) and \((G^{(0)},\cdots,G^{(\ell(N))})\).
	Therefore one needs to minimize the expectation of a deterministic function of
	the parameter \(\balpha\) in \(\mathbb{R}^p\) and the random
	vector \((G^{(0)},\cdots,G^{(\ell(N))})\).
	Such problems can be  solved numerically by classical stochastic optimization
	algorithms, 
	such as those introduced in the pioneering work of \citet{robbins_monro_51} and
	its extensions \citep[e.g.][]{chen_zhu_86}. 
	We refer the interested reader to the classical 
	references \citet{duflo_97,chen_2002,kushner_yin}.

	In our context \eqref{eq:optimizationparametric}, the Robbins-Monro algorithm 
	\eqref{eq:evoltionrobinsmonro}
	works as follows.
	Start with an arbitrary initial condition \(\balpha_0\) in \(\mathbb{R}^p\).
	At step \(\gamma+1\), given the current approximation \(\balpha_\gamma\) of the
	optimal value 
	\(\balpha^*\),
	simulate independent Gaussian random variables 
	\((G^{(0)}_{\gamma+1},\cdots,G^{(\ell(N))}_{\gamma+1})\) 
	and compute the terminal wealth \(W^{(\balpha_\gamma)} \).
	Then, update  the parameter \(\balpha\) by the induction formula
	\begin{equation}\label{eq evolution de alpha}
		\balpha_{\gamma + 1} = \balpha_{\gamma} - \rho_{\gamma+1}
		\nabla_{\balpha}\left[  S ( W^{(\balpha_\gamma)})  \right],
	\end{equation} 
	where \((\rho_\gamma)\) is a deterministic decreasing sequence.
	In addition, one can use an improvement of this algorithm due to
	\citet{chen_zhu_86}.
	Let \((K_l,l\geq 0)\) be an increasing sequence of compact sets such that
	\begin{equation}\label{eq:compactcroissant}
		K_l \subset \text{Int}(K_{l+1}) \quad \mbox{ and } \quad \lim_l K_l =
		\mathbb{R}^p,
	\end{equation}
	where \(\text{Int}(K_{l+1})\) denotes the interior of the set \(K_{l+1}\).
	The initial condition \(\balpha_0\) is assumed to be in \(K_0\) and we set 
	\(l(0)=0\).
	At each step \(\gamma\) in \eqref{eq evolution de alpha}, if
	\(\balpha_{\gamma+1}\in K_{l(\gamma)}\),
	we set \(l(\gamma+1) = l(\gamma)\) and go to step \(\gamma+1\). Otherwise, that
	is  
	if \(\balpha_{\gamma+1}\notin K_{l(\gamma)}\),
	we set \(\balpha_{\gamma+1}=\balpha_0\) and \(l(\gamma+1) = l(\gamma) + 1\).
	This modification avoids that the stochastic algorithm may blow up during the
	first steps
	and, from a theoretical point of view, allows to prove its convergence under
	weaker 
	assumptions than required for the standard Robbins-Monro method.
	
	\subsection{Summary of the method}
	\subsubsection*{Our setting}
	\begin{itemize}
		\item 
		The interest rate model satisfies: for all \(0\leq j < N\), there exist 
		an integer \(\ell(j)\) and a function \(\Phi_j\) such that
		\[
		(R_t,R_{T_0},\cdots,R_{T_j}) \overrel{\mathcal{L}}{=}
		\Phi_j(R_t,G^{(0)},\cdots,G^{(\ell(j))}).
		\]
		\item 
		For all \(0\leq j < N\), a finite truncation set
		\(\Lambda^{(j)}\subset\mathbb{N}^{\ell(j)+1}\) is given and
		\(\Lambda^{(-1)}=\{0\}\).
		\item A strategy \(\bpi=\bpi(\balpha)\) is defined by 
		\(\balpha = (
		\alphanjji
		, -1 \leq j<i\leq N-1, \bfn_{(j)} \in\Lambda^{(j)})\). More precisely, the
		number of zero-coupon bonds with maturity \(T_i\)
		bought or sold at time \(T_j\) is
		\begin{equation}
			\label{eq:tronc}
			\piji
			= \sum_{\bfn_{(j)}\in\Lambda^{(j)}}
			\alphanjji
			\prod_{m=0}^{\ell(j)}\dfrac{H_{n_m}(G^{(m)})}{\sqrt{n_m!}}, \quad \mbox{ for } i
			\leq N-1,
		\end{equation}
		and 
		\(\pijN\)
		is deduced from the self-financing equation
		\begin{equation}
			\label{eq:autofin}
			\sum_{i<j}
			\piij
			+
			P(j)
			= \sum_{i>j}\Psi_{j,i}(
			\piji
			)
		\end{equation}
		(one possibly needs to use a classical iterative procedure to solve this
		equation numerically).
		\item One is given an increasing sequence of compact sets \((K_l, l\geq 0)\)
		satisfying \eqref{eq:compactcroissant} and a  sequence 
		of parameters \((\rho_\gamma, \gamma\geq 1)\) decreasing to \(0\).
	\end{itemize}
	\subsubsection*{Our stochastic optimization algorithm}
	Assume that the parameter \(\balpha_{\gamma}= (
	\alphanjjik
	, -1\leq j<i\leq N-1,\bfn_{(j)} \in\Lambda^{(j)})\) and \(l_\gamma\) are given
	at step \(\gamma\). At step \(\gamma+1\):
	\begin{enumerate}
		\item Simulate a Gaussian vector
		\((G^{(0)}_{\gamma+1},\cdots,G^{(\ell(N))}_{\gamma+1})\).
		\item Deduce the quantities of zero-coupon bonds  from \eqref{eq:tronc} and
		\eqref{eq:autofin}:
		\[
		\pi_{\gamma+1}(j,i)
		= \sum_{\bfn_{(j)}\in\Lambda^{(j)}}
		\alphanjjik
		\prod_{m=0}^{\ell(j)}\dfrac{H_{n_m}(G^{(m)}_{\gamma+1})}{\sqrt{n_m!}}, \quad
		\mbox{ for } i \leq N-1,
		\]
		and get \(
		\pi_{\gamma+1}(j,N)
		\) from the self-financing equation
		\[
		\sum_{i<j}
		\pi_{\gamma+1}(j,i)
		+
		P(j)
		= \sum_{i>j}\Psi_{j,i}(
		\pi_{\gamma+1}(j,i)
		).
		\]
		\item Compute the terminal wealth
		\[
		W^{(\balpha_\gamma)}_{\gamma+1}  = \sum_{j<N}
		\pi_{\gamma+1}(j,N)
		+
		P(N).
		\]
		\item Update the parameters 
		\[
		\balpha_{\gamma + 1} = \balpha_{\gamma} - \rho_{\gamma+1}
		\nabla_{\balpha}\left[  S ( W^{(\balpha_\gamma)}_{\gamma+1} )  \right].
		\]
		\item If \(\balpha_{\gamma+1}\notin K_{l_\gamma}\), set
		\(\balpha_{\gamma+1}=\balpha_{0}\) and \(l(\gamma+1) = l(\gamma) + 1\).
		\item Go to 1 (or, in practice, stop after \(\Gamma\) steps).
	\end{enumerate}

	\subsection{Error analysis}\label{subsection error analysis}
	In this subsection we study the convergence (Theorem~\ref{th convergence}) and 
	convergence rate (Theorem~\ref{th rateofconvergence}) of the stochastic
	algorithm
	used in Step 2, when the total number of steps \(\Gamma\) tends to infinity.
	We introduce some notation. Recall \eqref{eq evolution de alpha} and write
	\begin{equation}\label{eq evolalphadecomp}
		\balpha_{\gamma + 1} = \balpha_{\gamma} -
		\rho_{\gamma+1}\mathbb{E}\left[\nabla_{\balpha}S(W^{(\balpha_\gamma)}_{\gamma+1}
		)\right]
		-\rho_{\gamma+1}\delta M_{\gamma+1} + \rho_{\gamma+1}p_{\gamma+1}.
	\end{equation}
	Here, \(\delta M_{\gamma+1}\) is given by
	\begin{equation}\label{eq incrementmart}
		\delta M_{\gamma+1} = \nabla_{\balpha}S(W^{(\balpha_\gamma )}_{\gamma+1} ) 
		- \mathbb{E}\left[\nabla_{\balpha}S(W^{(\balpha_\gamma)}_{\gamma+1} )\right].
	\end{equation}
	The last term \(\rho_{\gamma+1}p_{\gamma+1}\) in \eqref{eq evolalphadecomp}
	represents the \textit{reinitialization} of the algorithm 
	if \(\balpha_{\gamma+1}\notin\ K_{l(\gamma)}\) 
	i.e. \(p_{\gamma+1}\) is fixed such that 
		\(\balpha_{\gamma+1}=\balpha_0\).
	
	Let us now recall the convergence theorem obtained by 
	\citet[Theorem~1]{lelong-2008} in our setting.
	\begin{theorem}\label{th convergence}
		Assume
		\begin{enumerate}
			\item[(A1)] The function \(\balpha \mapsto \mathbb{E}[S(W^{(\balpha)})]\)
			is strictly concave or convex,
			\item[(A2)] \(\sum_\gamma \rho_\gamma =\infty\), \(\sum_\gamma \rho_\gamma^2 <
			\infty\),
			\item[(A3)] The function \(\balpha \mapsto \mathbb{E}[
			\|\nabla_{\balpha} 
			S(W^{(\balpha)})\|^2
			]\) is bounded on compact sets.
		\end{enumerate}
		Then the sequence \((\balpha_\gamma,\gamma\geq 1)\) converges a.s. to the unique
		
		optimal parameter \(\balpha^*\) such that
		\[
		\inf_{\balpha\in\mathbb{R}^p}\mathbb{E}\left[ S\left(W^{(\balpha)}\right)\right]
		= \mathbb{E}\left[ S\left(W^{(\balpha^*)}\right)\right]. 
		\]
	\end{theorem}
	
	Hypothesis~\ref{hypo:regularitepsi} and \eqref{eq:tronc} imply that
	(A3) is satisfied. Before giving examples of situations where  
	(A1) is fulfilled, let us check that \(W^{(\balpha)}\) is
	a concave function of \(\balpha\).
	%\end{document} %OK
	
	\begin{prop}\label{prop convexite}
		The terminal wealth \(W^{(\balpha)}\) is a concave function of the parameter
		\(\balpha\).
	\end{prop}
	\begin{proof}
		Recall that the terminal wealth is given by  \eqref{eqn_final_wealth}.
		The payoff of the swap \(P(N)\)
		does not depend 
		on \(\balpha\). We only have to deal with the 
		quantities \(
		\pi(i,N)
		\) of the zero-coupon bonds
		with maturity \(T_N\) bought at time \(T_i\).
		%They solve the self-financing equation 
		They satisfy the self-financing equation \eqref{eq selffinancing} and thus
		\begin{equation}\label{eq piin}
			\pi(i,N) =  \Psi_{i, N}^{(-1)}\left(
			\sum_{j < i} \piji  -  \sum_{i < j < N} \Psi_{i, j}\left( \piij\right)+   
			P(i)
			\right),
		\end{equation}
		where \(\Psi_{i,N}^{(-1)}\) is the inverse of the price function \(\Psi_{i,N}\) 
		(see \eqref{eq psiij}).
		Moreover, quantities \(\piij\), \(i<j\leq N-1\)
		are linear in \(\balpha\) (see \eqref{decompositionL2}).
		
		Recall that \(\Psi_{i,j}\) is convex, thus \(-\Psi_{i,j}\) is concave and
		the argument in \eqref{eq piin} is a 
		concave function of \(\balpha\). Finally \(\Psi_{i,N}^{(-1)}\) is an increasing
		concave function,
		from which  \(\pi(i,N)\) is a concave function of \(\balpha\).
	\end{proof}
	
	%\end{document} %OK
	
	The preceding observation shows that (A1) is satisfied
	when  \(S\) is a utility function (and thus
	increasing and concave) and satisfies \(S(0) = 0\). 
	Notice that the optimization problem~(\ref{eq:optimizationparametric})
	then penalizes the losses and promotes the gains.
	In Sections~\ref{sec numericalwithoutliquidity} and \ref{sec
		numericalwithliquidity} we will see another situation where 
	Theorem~\ref{th convergence} applies.
	
	%\end{document} %OK
	
	Given suitable functions $S$,
	Theorem~\ref{th convergence} guarantees the convergence of our algorithm
	towards the optimal 
	parameters. The following theorem provides the rate of convergence
	\citep{lelong2013}.
	%\end{document}%OK
	\begin{theorem}\label{th rateofconvergence}
		Let
		\begin{equation}\label{eq formedegamman}
			\rho_\gamma := \dfrac{v_1}{\left(v_2 + \gamma \right)^\beta},
		\end{equation} 
		for some positive \(v_1\), \(v_2\) and \(\beta\in(1/2,1)\).
		Denote by \(\Delta_\gamma\) the normalized centered error
		\[
		\Delta_\gamma = \dfrac{\balpha_\gamma-\balpha^*}{\sqrt{\rho_\gamma}}.
		\]
		Assume
		\begin{enumerate}
			\item[(A1)]The function \(\balpha \mapsto \mathbb{E}[S(W^{(\balpha)})]\)
			is concave or convex.
			\item[(A4)] For any \(q >0\), the series 
			\[
			\sum_\gamma\rho_{\gamma+1}\delta
			M_{\gamma+1}\Ind_{\{|\balpha_\gamma-\balpha^*|\leq q\}}
			\]
			converges almost surely.
			\item[(A5)]There exist two real numbers \(A_1>0\) and \(A_2>0\) such that 
			\[
			\sup_\gamma\mathbb{E}\left[|\delta
			M_{\gamma}|^{2+A_1}\Ind_{\{|\balpha_\gamma-\balpha^*|\leq A_2\}}\right]<\infty.
			\]
			\item[(A6)] There exists a symmetric positive definite matrix \(\Sigma\) such
			that
			\[
			\mathbb{E}\left[\delta M_\gamma\delta
			M_\gamma^\mathsf{t}\middle|\mathcal{F}_{\gamma-1}\right]\Ind_{\{|\balpha_{\gamma-1}-\balpha^*|\leq
				A_2\}}
			\mathrel{\mathop{\kern 0pt \longrightarrow}\limits_{\gamma\rightarrow
					\infty}^{\mathbb{P}}}\Sigma.
			\]
			\item[(A7)] There exists \(\mu>0\) such that \(\forall n\geq 0\),
			\(d(\balpha^*,\partial K_n)\geq \mu\),
		\end{enumerate}
		where \(\partial K_n\) denotes the boundary of \(K_n\).
		Then, the sequence \((\Delta_\gamma,\gamma\geq 1)\) converges in distribution to
		a normal random variable 
		with mean \(0\) and  covariance matrix linearly depending  on \(\Sigma\).
	\end{theorem}
	%\end{document} %PAS OK
	\begin{rem}\label{rem:rem1}
		As explained in detail in \citet[Sec. 2.4]{lelong2013}, 
		the assumptions of Theorem~\ref{th rateofconvergence} are satisfied as soon 
		as 
		\begin{itemize}
			\item There exists \(A_3>0\) such that  
			\[
			\forall C>0, \quad \mathbb{E}\left[\sup_{|\balpha|\leq
				C}\left|S\left(W^{(\balpha)}\right)\right|^{2+A_3}\right] < \infty.
			\] 
			\item The function \(\balpha \mapsto \mathbb{E}[S(W^{(\balpha})]\)
			is strictly concave or convex.
		\end{itemize}
		The first condition is usually easily satisfied, e.g. when \(S\) has a
			polynomial growth
			at infinity because the moments of the wealth process are typically finite.
			The both properties are e.g. fulfilled by the example studied in
			Sections~\ref{sec numericalwithoutliquidity} and \ref{sec
				numericalwithliquidity} (see the discussions at the beginning of these
			sections). 
	\end{rem}
	\subsection{Performance of the optimal truncated strategy without liquidity
		cost}\label{subsec35}
	
	The numerical error on
	the optimal wealth decreases when the \(\Lambda^{(j)}\)'s tend to
	\(\mathbb{N}^{\ell(j)+1}\).
	In this subsection, we provide a theoretical estimate on the error resulting 
	from the truncation in \eqref{decompositionL2} in the idealistic 
	context of no liquidity cost and general Gaussian affine models
	\citep{dai_singleton_2000}.

	In \citet{dai_singleton_2000}, general Gaussian affine models are introduced 
	for which, for any two times \(s < t\), there exist standard
	independent Gaussian random variables \(G^{(0)},\cdots,G^{(M)}\) and
	real numbers \(\mu, \lambda_0, \cdots,\lambda_M\)
	such that the prices of zero-coupon bonds have the form
	\[
	B(s,t) = \exp\left(-\mu - \lambda_0 G^{(0)} - \cdots - \lambda_M G^{(M)}\right).
	\]
	A control of the error of truncation is given in the following proposition.
	\begin{prop}
		In the above context, if the truncation set defined in \eqref{eq:truncation} is
		\(\Lambda^{(j)}:=\{n_0+\cdots+n_{\ell(j)} \leq d\}\), then 
		\begin{equation}
			\mathbb{E}
			\left(W^{(\balpha^{*})} \right)^2
			\leqslant
			C_0
			\dfrac{C_1^{d + 1}}{(d + 1)!},
		\end{equation}
		where \(C_0\) and \(C_1\) are some positive constants.
	\end{prop}
	The proposition is a straightforward consequence of 
	\eqref{eqn_final_wealth} and the next lemma applied to \(X=\bpi(\balpha^*)\).
	This lemma also 
	allows one to precise the values of \(C_0\) and \(C_1\).
	\begin{lemma}\label{lemma erreur troncature}
		Consider the random variable 
		\[
		X := \exp\left(\mu + \sum_{m=0}^{M}\lambda_mG^{(m)}\right),
		\]
		where \(\mu, \lambda_0, \cdots,\lambda_M\) are real numbers, and \(G^{(0)},
		\cdots,G^{(M)}\)
		are independent standard Gaussian random variables. Consider the projection 
		\(X^{d}\) of
		\(X\) on the subspace of \(L^2(G^{(0)}, \cdots,G^{(M)})\) 
		generated by
		\[
		\left(\prod_{m=0}^M \dfrac{H_{n_m}(G^{(m)})}{\sqrt{n_m !}}, \quad n_0  + \cdots
		+ n_M  \leq d\right).
		\]
		We have
		\begin{equation}
			\label{eq distance entre pistar et proj}
			\|X - X^{d}  \|_2
			\leqslant 
			\exp\left(\mu +\lambda_0^2+\cdots\lambda_M^2\right)
			\dfrac{\left(\lambda_0^2 + \cdots +
				\lambda_{M}^2\right)^{\tfrac{d+1}{2}}}{\sqrt{(d+1)!}}.
		\end{equation}
	\end{lemma}
	We postpone the proof of this lemma to the Appendix.

	\section{Numerical validation of the optimization procedure: an example without
		liquidity cost}\label{sec numericalwithoutliquidity}
	In this section we study the accuracy of our algorithm in
	the no liquidity cost case where a perfect replication strategy is known
	(see Section~\ref{sub:rappelsswap}). We minimize the quadratic risk measure
		of the hedging error.
	
	The bond market model 
	is the Vasicek model which is the simplest Gaussian model:
	\begin{equation}\label{eq shortrate}
		dR_\theta = A(r_\infty-R_\theta) d\theta + \sigma dB_\theta,
	\end{equation}
	where 
	\(A\) is the mean reverting rate,  
	\(r_\infty\) is the mean of the equilibrium measure,
	\(\sigma\) is the volatility and
	\((B_\theta,\theta\geq 0)\)  is a one-dimensional Brownian motion.
	
	Notice that
	\begin{equation}\label{eq:incrementgauss}
		\forall u < v,\quad R_v = r_\infty + (R_u - r_\infty)e^{-A(v-u)} + \sigma
		e^{-A(v-u)}
		\int_u^v e^{A(\theta-u)} dB_\theta.
	\end{equation}
	Therefore, there exists an \textit{i.i.d.} sequence \((G^{(0)}, G^{(1)}, \cdots,
	G^{(N)})\)
	of \(\mathcal{N}(0,1)\)
	Gaussian random variables such that 
	\begin{equation}\label{eq:tauxinteretsfngaussienne}
		\forall k=0,\cdots, N, \quad R_{T_k} = \Phi(T_k-T_{k-1},R_{T_{k-1}},G^{(k)}),
	\end{equation}
	where
	\[
	\Phi(\eta,r,G) = r_\infty + 
	e^{-A\eta}
	(r - r_\infty) + G
	\sigma \sqrt{\dfrac{1 - 
			e^{-2A\eta}
		}{2A}}.
	\]
	In our numerical experiments, we have chosen the 
	following typical values of the parameters 
	\(A=10\%\), \(r_\infty = 5\%\), \(\sigma = 5\%\).
	With this choice of parameters,
	the mean yearly interest zero-coupon rates with maturity less than 10 years 
	take values between 
	\(3\%\) and \(5\%\).
	
	Our numerical study concerns the minimization of
	the quadratic mean hedging error which corresponds to the choice
	\(S(x) = x^2\) in \eqref{eq:optimizationparametric}. This choice penalizes
	gains and losses in a symmetric way and aims to construct a strategy as close as
	possible to the
	exact replication strategy.
	
	In the no liquidity cost case, the terminal wealth \(W^{(\balpha)}\) is
	a linear function of the parameter \(\balpha\)
	and therefore assumption (A1) of Theorem~\ref{th convergence}
	is obviously satisfied.
	
	Given a degree of truncation \(d\), the set \(\Lambda^{(j)}_d\) is chosen as
	\begin{equation}
		\label{eq ensemble de troncation}
		\Lambda^{(j)}_d :=
		\{\bfn_{(j)}=(n_0,\cdots,n_j)\in\mathbb{N}^{j+1},n_0+\cdots+n_j\leq d\}.
	\end{equation}
	We have to optimize the real-valued parameters
	\(\alphanjji\)
	for \(j<i\) and \(\bfn_{(j)}\in\Lambda^{(j)}_d\).
	The quantities of zero-coupon bonds to exchange
	are given by \eqref{eq:truncation}.
	The choice of the sequence \((\rho_\gamma, \gamma\geq 1)\) 
	in \eqref{eq evolution de alpha} is crucial. Choose $\rho_\gamma$
	as in~(\ref{eq formedegamman}).
	We discuss the sensitivity of the method to the parameters $v_1$, $v_2$, $\beta$
	in Section~\ref{SectionSequence}.
	We also discuss the sensitivity of the results 
	to the number \(\Gamma\) of steps.
	
	In all the sequel, we use the following notation.
	
	\noindent\textbf{Notation}
	For all vector \(\balpha = (\alphanjji, -1\leq j<i\leq N-1)\), we set 
	\begin{equation}
		\label{eq:deffonctionvaleur}
		v(\balpha):=\mathbb{E}\left[W^{(\balpha)}\right]^2,
	\end{equation}
	where the expectation is computed only with respect to the Gaussian 
	distribution \((G^{(0)},\cdots,G^{(\ell(N))})\).
	\subsection{Empirical study of the truncation errors (Step~1)}
	\label{SectionTruncationSpace}
	
	In this subsection we develop an empirical validation of the projection step
	presented in Section \ref{subsection:finitiedimensionalprojection}
	
	We observe that the quadratic mean hedging error decreases very fast to \(0\)
	when the
	degree of  truncation increases.
	For a notional equal to \(1\), the error 
	is of the order 
	of one basis point 
	(a hundredth of percent) for a degree \(d=3\) and a small number of dates \(N\),
	and for \(d=4\) and
	for larger values of \(N\).
	
	Figure~\ref{truncationPoly} below shows 
	\(v(\balpha^{*,d})\), with the optimal parameter \(\balpha^{*,d}\)
	corresponding to the truncation set \eqref{eq ensemble de troncation}.
	We have used the explicitly known 
	finite dimensional projections of the optimal strategies
	without liquidity cost to obtain \(\balpha^{*,d}\), and a Monte Carlo procedure
	to
	compute \(v\).
	Table~\ref{tabletroncature} shows some  values 
	used to plot Figure~\ref{truncationPoly}.

	\begin{figure}[ht]
		\begin{center}
			\includegraphics[width=0.7\textwidth,angle=-90]{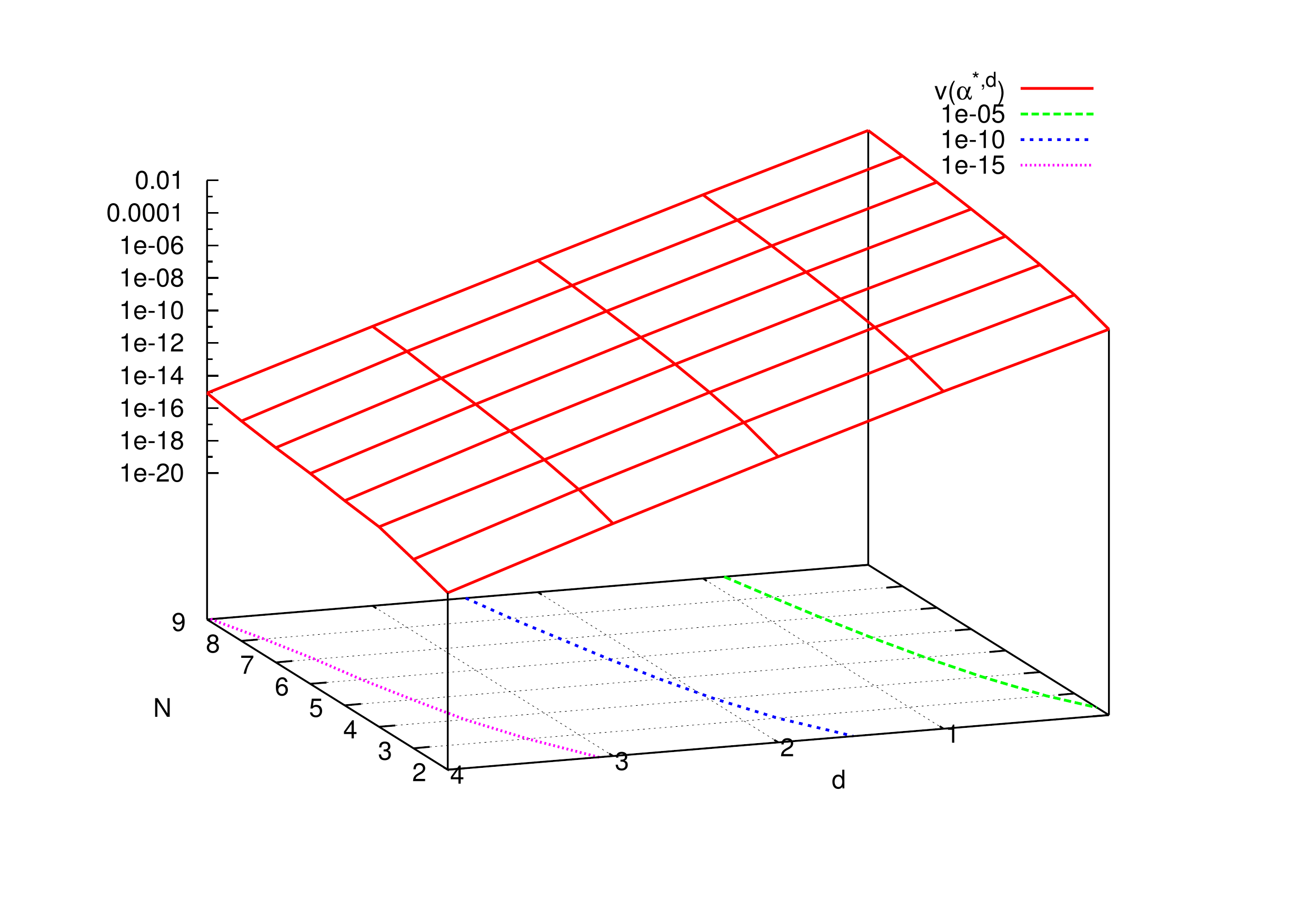}
		\end{center}
		\caption{
			\label{truncationPoly}}
	\end{figure}
	
	\begin{table}[ht]\footnotesize
		\centering
		\begin{tabular}{|c|c|c|}
			\hline
			degree d & N=2 & N=3 \\
			\hline
			0 & 5.2 E-6 & 3.0 E-5 \\
			\hline
			1 & 5.4 E-9 & 3.1 E-8\\
			\hline
			2 & 3.7 E-12 & 2.0 E-11 \\
			\hline
			3 & 1.9 E-15 & 1.9 E-14\\
			\hline
			4 &2.2 E-18 & 3.9 E-15 \\
			\hline
		\end{tabular}
		\caption{\(v(\balpha^{*,d})\)\label{tabletroncature} }
	\end{table}

	\subsection{Empirical study of the  optimization step (Step 2)}
	The stochastic algorithm converges almost surely
	to the optimal coefficient \(\balpha^*\). In this part, we 
	empirically study  the convergence rate in terms of the number of 
	steps \(\Gamma\) and the choice of the sequence \((\rho_\gamma)\).

	\subsubsection{A typical evolution of \protect{$(\alpha_\gamma)$}}
	In this subsection, we consider a swap with two payment dates (\(N=2\)).
	We consider the truncation set \(\Lambda^{(0)} = \{0,1\}\).
	The objective is to approximate 
	\(
	\balpha^*=
	(\alpha^{0,*}(-1,0),\alpha^{0,*}(-1,1), \alpha^{0,*}(0,1),\alpha^{1,*}(0,1)).
	\)
	In Figure~\ref{fig evolution des parametres},
	the four parameters 
	\(\balpha=(
	\alpha^{0}(-1,0),\alpha^{0}(-1,1),\alpha^{0}(0,1),\alpha^{1}(0,1)
	)\)
	evolve according to \eqref{eq evolution de alpha}
	where the sequence \((\rho_\gamma,\gamma\geq 1)\) is defined by
	\eqref{eq formedegamman} with \(v_1 = 10^7\), \(v_2=1\) and
	\(\beta = 1\).
	
	As expected, the sequence \((\balpha_\gamma)\) converges to \(\balpha^*\). 
	However, the evolution is 
	quite \textit{slow} although we have 
	empirically chosen the parameters \(v_1\), \(v_2\) and \(\beta\)  
	in a favorable way.
	
	In Figure~\ref{figOptPath}, we plot (in violet) 
	\(v( \alpha^{0}(-1,0),\alpha^{0}(-1,1),\alpha^{0,*}(0,1),\alpha^{1,*}(0,1))\)
	as a function of \(\alpha^{0}(-1,0)\) and \(\alpha^{0}(-1,1)\). 
	We also plot in green the path \((v(\balpha_\gamma),0\leq \gamma \leq \Gamma)\).
	The figure shows that after \(\Gamma=10000\) steps the hedging error is small 
	though the optimal parameters have not been approximated accurately 
	(notice that the violet surface is flat).
	\begin{figure}[ht]
		\begin{center}
			\includegraphics[width=0.7\textwidth]{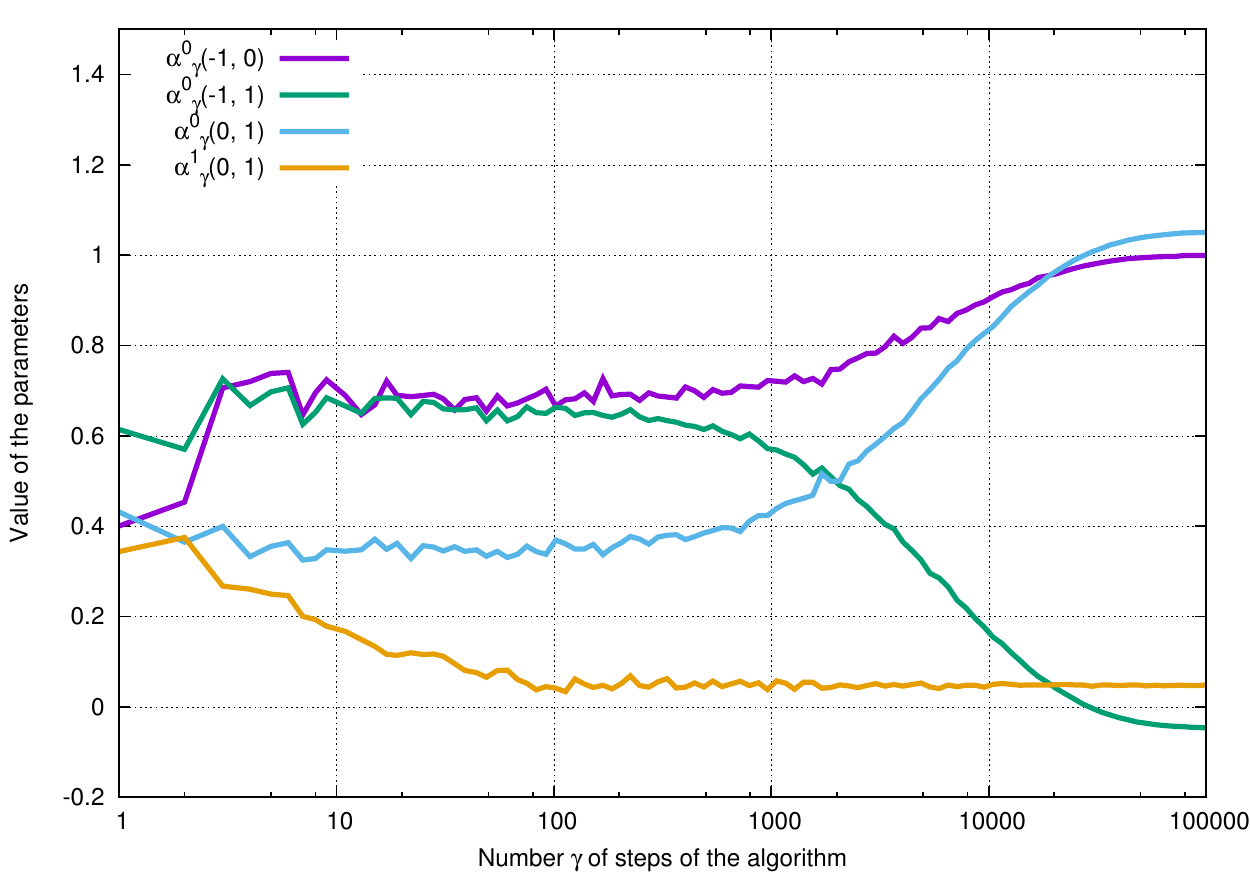}
		\end{center}
		\caption{Evolution of the parameters \(\alpha^0_\gamma(-1,0)\), 
			\(\alpha^0_\gamma(-1,1)\),
			\(\alpha^0_\gamma(0,1)\),
			\(\alpha^1_\gamma(0,1)\)
			in terms of \(\gamma\). }\label{fig evolution des parametres}
	\end{figure}

	\begin{figure}[ht]
		\begin{center}
			\includegraphics[angle=-90,
			width=0.9\textwidth]{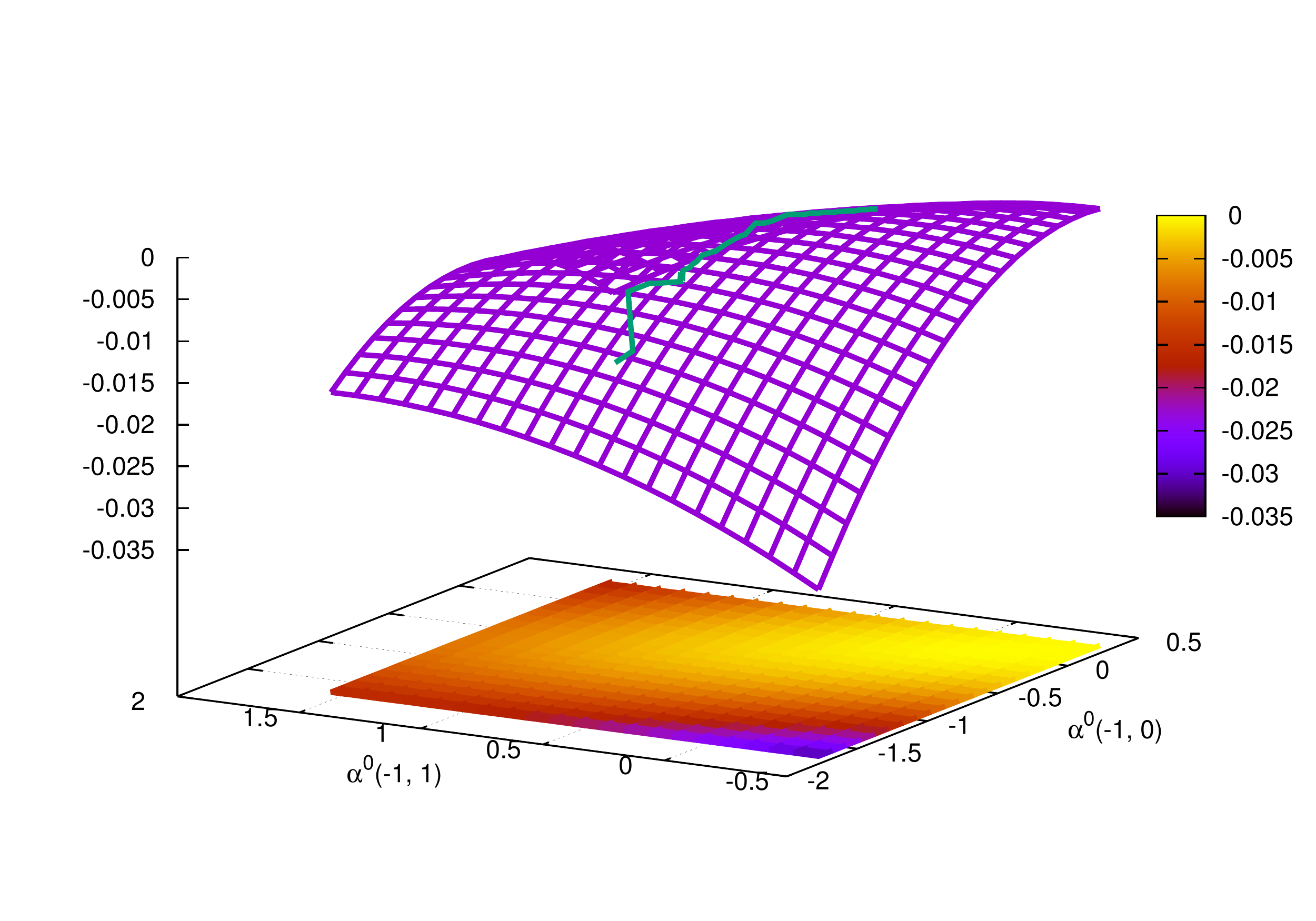}
		\end{center}
		\caption{
			Violet surface: \( -
			v(\centerdot,\centerdot,\alpha^{0,*}(0,1),\alpha^{1,*}(0,1))\).
			Green curve: evolution of 
			\((\alpha^0\gamma(-1,0), 
			\alpha^0_\gamma(-1,1),-
			v(\alpha^{0}_\gamma(-1,0),\alpha^{0}_\gamma(-1,1),\alpha^{0,*}(0,1),\alpha^{1,*}(0,1))\)
			in terms of \(\gamma=1,\cdots,10000\)
			(\(v_1=10\), \(v_2=1\) and \(\beta=0.6\)).}\label{figOptPath}
	\end{figure}

	\clearpage
	
	\subsubsection{Sensitivity to the choice of the sequence
		\protect{$(\rho_\gamma)$}}
	\label{SectionSequence}
	
	Theorem~\ref{th convergence} states the convergence of
	the optimization method for all sequence \((\rho_\gamma)\)
	satisfying (A2). We here study 
	the sensitivity of the results to the parameters 
	\(v_1\), \(v_2\), \(\beta\) of
	sequences of type \eqref{eq formedegamman} and
	to the  total number of steps \(\Gamma\).
	
	Tables~\ref{table fonctionv1} and \ref{table fonctionv1b}
	show the expected value function 
	obtained after \(\Gamma=10E4\), \(10E5\) and \(10E6\) steps. The 
	expected value function is 
	estimated by means of 
	a classical Monte Carlo procedure.
	Table~\ref{table pas constant} shows the same results 
	with a sequence \(\rho_\gamma = v_1\)
	which does not satisfy condition (A2).

	\begin{table}[ht]\footnotesize
		\begin{center}
			\begin{tabular}{|c||c|c|c|c|c|c|c|c|}
				\hline
				\backslashbox{\(\Gamma\)}{\(v_1\)} & 1 & 10 & 100 & 1000 & 10 000 &  20 000 &
				10E5 & 10E6\\
				\hline
				10E4 & 7.2 E-4 & \color{red}{2.6 E-6} & \color{red}{1.1 E-6} & \color{red}{8.9
					E-7} & 5.1 E-4 &  8.3 E-4 & 1.0 E-1 & 1.3 E-1 \\
				\hline  
				10E5 &1.4 E-5 & 9.3 E-7 & 8.9 E-7 & 4.7 E-7 & \color{red}{2.3 E-8} 
				& 1.7 E-2 & 9.1 E+3 & 9.5 E+5\\
				\hline
				10E6 & 3.8 E-6 & 9.2 E-7 & 7.7 E-7 & 1.5 E-7 & \color{red}{6.6 E-12}  &
				\color{red}{5.4 E-12} & 1.0 E-5 & 15.8\\
				\hline
			\end{tabular}
		\end{center}
		\caption{\(v(\balpha_\Gamma)\) (\(v_2 = 1000\), \(\mathbf{\beta = 0.6}\)) }
		\label{table fonctionv1}\normalsize
	\end{table}

	\begin{table}[ht]\footnotesize
		\begin{center}
			\begin{tabular}{|c||c|c|c|c|c|c|c|}
				\hline
				\backslashbox{\(\Gamma\)}{\(v_1\)} & 1 & 10 & 100 & 1000 & 10 000 & 13 000 & 2E4
				\\
				\hline
				10E4  & 8.8 E-3 & 1.0 E-3 & 7.8 E-6 & \color{red}{9.8 E-7} & \color{red}{8.8
					E-7} & \color{red}{1.7 E-6} & \color{red}{9.7 E-7} \\
				\hline
				10E5  & 6.6 E-3  & 1.1 E-4 & \color{red}{9.3 E-7} & \color{red}{9.0 E-7} &
				\color{red}{6.8 E-7} & \color{red}{5.5 E-7} & \color{red}{4.8 E-7}\\
				\hline
				10E6  & 4.5 E-3 & 1.8 E-5 & 9.3 E-7 & 8.7 E-7 & 5.1 E-7 & 3.6 E-7 & 2.9 E-7 \\
				\hline  
			\end{tabular}
			\begin{tabular}{|c||c|c|c|c|c|c|c|}
				\hline
				\backslashbox{\(\Gamma\)}{\(v_1\)} &  1E5 & 5E5 & 1E6 & 2E6 & 3E6 & 4E6 & 5E6\\
				\hline
				10E4  & 1.3 E-3 & 3.3 E-2 & 7.3 E-1 & 4.5 E+1 & 1.1 E-2 & 6.5 E+2 & 6.9 E+1 \\
				\hline
				10E5  & \color{red}{6.6 E-7} & 1.2 E-5 & 4.1 E-4 & 2.7 E-1 & 8.1 E-3 & 1.7 E-2 &
				1.4 E+1\\
				\hline
				10E6  & 6.8 E-8 & 1.1 E-10 & \color{red}{5.3 E-12} & \color{red}{7.4 E-12} &
				\color{red}{7.0 E-12} & 2.6 E-5 & 1.5 E-6\\
				\hline
			\end{tabular}
		\end{center}
		\caption{\(v(\balpha_\Gamma)\) (\(v_2 = 1000\), \(\mathbf{\beta = 0.9}\)) }
		\label{table fonctionv1b}\normalsize
	\end{table}
	
	\begin{table}[ht]\footnotesize
		\begin{center}
			\begin{tabular}{|c|c|c|c|c|c|c|c|c|c|}
				\hline
				\backslashbox{\(\Gamma\)}{\(v_1\)} & 1 &  2 & 4 & 6 & 8 & 10 & 12 & 20\\
				\hline
				10E4  & 7.6 E-7 & 6.9 E-7 & 2.5 E-7 & 2.9 E-7 &
				1.4 E-6 & 1.9 E-7 & 7.6 E-7 & 6.1 E-6\\
				\hline
				10E5  & 6.9 E-7 & 6.8 E-7 & 4.1 E-7 & \color{red}{1.3 E-7} & 
				\color{red}{2.9 E-7} 
				& 3.2 E-7 & 3.9 E-6 & 3.2 E-4 \\
				\hline
				10E6  & 3.0 E-8 & 1.0 E-9 & \color{red}{5.1 E-12} & \color{red}{4.3 E-12} &
				\color{red}{5.1 E-12} & 
				\color{red}{4.2 E-12} & \color{red}{5.9 E-12} & 7.8 E-6\\
				\hline     
			\end{tabular}
		\end{center}
		\caption{\(v(\balpha_\Gamma)\) for a constant sequence \(\rho_\gamma = v_1\)}
		\label{table pas constant}\normalsize
	\end{table}

	We observe that the efficiency of the algorithm depends
	on the choice of the parameters 
	\(v_1\), \(v_2\), \(\beta\) and is really
	sensitive to it when the total number of steps \(\Gamma\) 
	is small.
	
	When \(\Gamma\) becomes large (e.g \(\Gamma=10E6\)), then the algorithm
	may seem to diverge if \(\beta\) is chosen 
	%careless. 
	carelessly.
	In fact,
	as the sequence \((\rho_\gamma,\gamma\geq 1)\) satisfies hypothesis (A2) of
	Theorem~\ref{th convergence},
	the algorithm converges to the optimal parameters but it is  far from
	\(\balpha^*\) after \(10E6\) steps.
	However, for each value \(\beta\), some \(v_1\)
	reduces the mean square hedging error  to 5 E-12.

	\section{An empirical study of the bid-ask spread costs impact}
	\label{sec numericalwithliquidity}
	We here present numerical results corresponding to two piecewise linear
	liquidity cost functions \(\Psi\): 
	\begin{align}
		\Psi^1_\lambda(T,U,\pi) &= (1 + \lambda  \,\text{sign}(\pi))B(T,U)\pi \label{eq
			liquidity 1}\\
		\Psi^2_{\lambda,C}(T,U,\pi) &= 
		\begin{cases} B(T,U)\pi \mbox{ for } |\pi|\leq C\\
			B(T,U)(C + (1+\lambda)(\pi - C) \mbox{ for } \pi > C\\
			B(T,U)(- C + (1-\lambda)(\pi + C) \mbox{ for } \pi < - C.
		\end{cases}
		\label{eq liquidity 2}
	\end{align}
	
	Despite the fact that we know there is no perfect hedging strategy
	in this context, we suppose the holder receives a null cash at time 
	\(t\) (which is the price of the swap in a no liquidity cost market).

		We now shortly check the convergence of the algorithm. Given piecewise linear
		cost functions \(\Psi\), 
		it is easy to prove that the terminal wealth \(W^{(\balpha)}\) is 
		piecewise linear in \(\balpha\) (see the proof of proposition~\ref{prop
			convexite}).
		Therefore, assumption (A1) of Theorem~\ref{th convergence}.
		However, 
		%\eot{condition (A1) in 
		Theorem~\ref{th convergence} 
		%	is satisfied since .....}
		does not apply
		to our context since \(\Psi^1_\lambda\) and \(\Psi^2_\lambda\)
		are piecewise linear and therefore are not
		continuously differentiable everywhere.
		Replace  \(\Psi^1_\lambda\) and \(\Psi^2_\lambda\) 
		by smooth approximations obtained by convolutions with kernels of the type
		\(1/\sqrt{2\pi\varepsilon}\exp(-x^2/(2\varepsilon))\), 
		\(\varepsilon\) small.
		Let \(\balpha^*_\varepsilon\) be the unique optimal parameter corresponding to 
		the new cost functions (existence and uniqueness
		of \(\balpha^*_\varepsilon\) are provided by Theorem~\ref{th convergence}).
		In view of \citet[Th~7.33]{Rockafellar_Tyrrell98}
		\((\balpha^*_\varepsilon)\) tends to \(\balpha^*\) when \(\varepsilon\)
		tends to \(0\).
		
		The preceding consideration to prove convergence is more theoretical than
		practical: in practice,
		the numerical results do not differ when  \(\varepsilon\) is small or 
		\(\varepsilon\) is null.
		
		Finaly, notice that the piecewise linearity in \(\balpha\) of \(W^{(\balpha)}\)
		implies that conditions in Remark~\ref{rem:rem1}
		are satisfied, so that Theorem~\ref{th rateofconvergence} precises the
		convergence rate.
	
	\subsection{Taking liquidity costs into account is really
		necessary}\label{sub:50}
	Consider two different strategies:
	(i) the strategy corresponding to the optimal parameters \(\balpha_0\) in the
	idealistic model without liquidity costs  and
	(ii) the null strategy \(\boldsymbol{\delta}_0\) defined as
	\[
	\boldsymbol{\delta}_0^{\bfn_{(j)}}(j,i)=0, \quad \mbox{ for all }-1 \leq j < i
	\leq N-1
	\mbox{ and } \bfn_{(j)}\in\Lambda^{(j)}.
	\]
	To satisfy the self-financing assumption \eqref{eq selffinancing}, at time
	\(T_j\) the payoff 
	\(P(j)\) of the swap \eqref{eq payoff of the swap} is used to buy  zero-coupon
	bonds with maturity
	\(T_N\). 
	
	Figure~\ref{fig performance avant optimisation} shows
	\(-v(\balpha_0)\) and \(-v(\boldsymbol{\delta}_0)\)
	in terms of the parameter \(\lambda\)
	where the cost function \(\Psi\) is as in \eqref{eq liquidity 1}.
	The mean square hedging error dramatically increases when, in the presence of
	liquidity costs,  
	the trader uses the strategy which is optimal in the no liquidity cost context.
	When the liquidity cost \(\lambda\) is larger
	than \(4\%\), it is even worse to use this strategy than to use the
	\(\boldsymbol{\delta}_0\) strategy!
	\begin{figure}[ht]
		\begin{center}
			\includegraphics[angle=-90,width=0.9\textwidth]{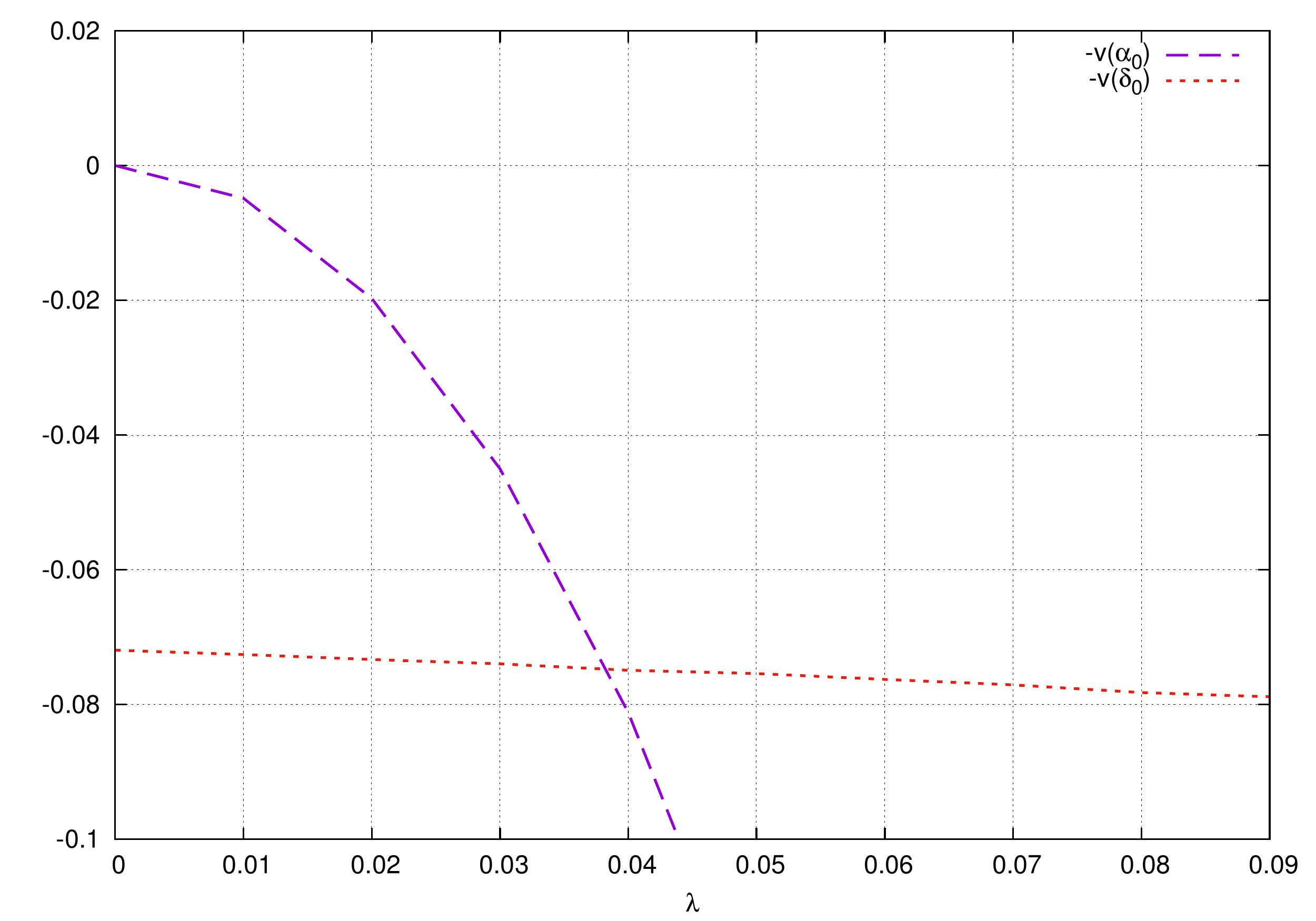}
			\caption{(short dashes) \(-v(\boldsymbol{\delta}_0)\)
				(long dashes) \(-v(\balpha_0)\)
				in terms  of the liquidity cost \(\lambda\) 
			}
			\label{fig performance avant optimisation}
		\end{center}
	\end{figure}
	
	\subsection{Probability distribution of the hedging error in the case~\eqref{eq
			liquidity 1}}
	In this section, the liquidity cost function \(\Psi\) is chosen as in
	\eqref{eq liquidity 1}. 
	
	After \(\Gamma\) steps of the stochastic optimization procedure with a sample
	\(\omega\)
	of the Gaussian vector
	\(((G^{(0)}_\gamma,\cdots,G^{(N)}_\gamma),\gamma=1,\cdots,\Gamma)\), 
	one obtains a random approximation \(\balpha_\Gamma(\omega)\)
	of the optimal parameter \(\balpha^*\).
	
	Figure~\ref{figlesdensites} shows the probability distribution of 
	the random variable \(v(\balpha_\Gamma(\omega))\) for \(\Gamma= 10000\)
	and Table~\ref{tabletroncature2} shows its mean and standard deviation
	for different values of \(\lambda\).
	\begin{table}[h]\footnotesize
		\begin{center}
			\begin{tabular}{|c|c|c|c|c|c|c|c|c|c|c|}
				\hline
				$\lambda$ & 0 & 0.01 & 0.02 & 0.03 & 0.04 & 0.05 & 0.06 & 0.07 & 0.08 &
				0.09 \\
				\hline
				\mbox{Mean}& 2.5E-32 & 0.0031 & 0.012 & 0.024 & 0.038 & 0.049 & 0.058 &
				0.065 & 0.11 & 0.12 \\
				\hline
				\mbox{Std dev.} & 6.7E-33 & 4.2E-5 & 2.8E-4 & 1.7E-3 & 0.016 & 0.017 &
				0.046 & 0.081 & 1.1 & 1.3 \\
				\hline
			\end{tabular}
		\end{center}
		\caption{
			Empirical mean and standard deviation of \(v(\balpha_{10000})\).}
		\label{tabletroncature2}
	\end{table}
	\normalsize 
	\begin{figure}[ht]
		\begin{center}
			\includegraphics[height=8cm,
			angle=-90]{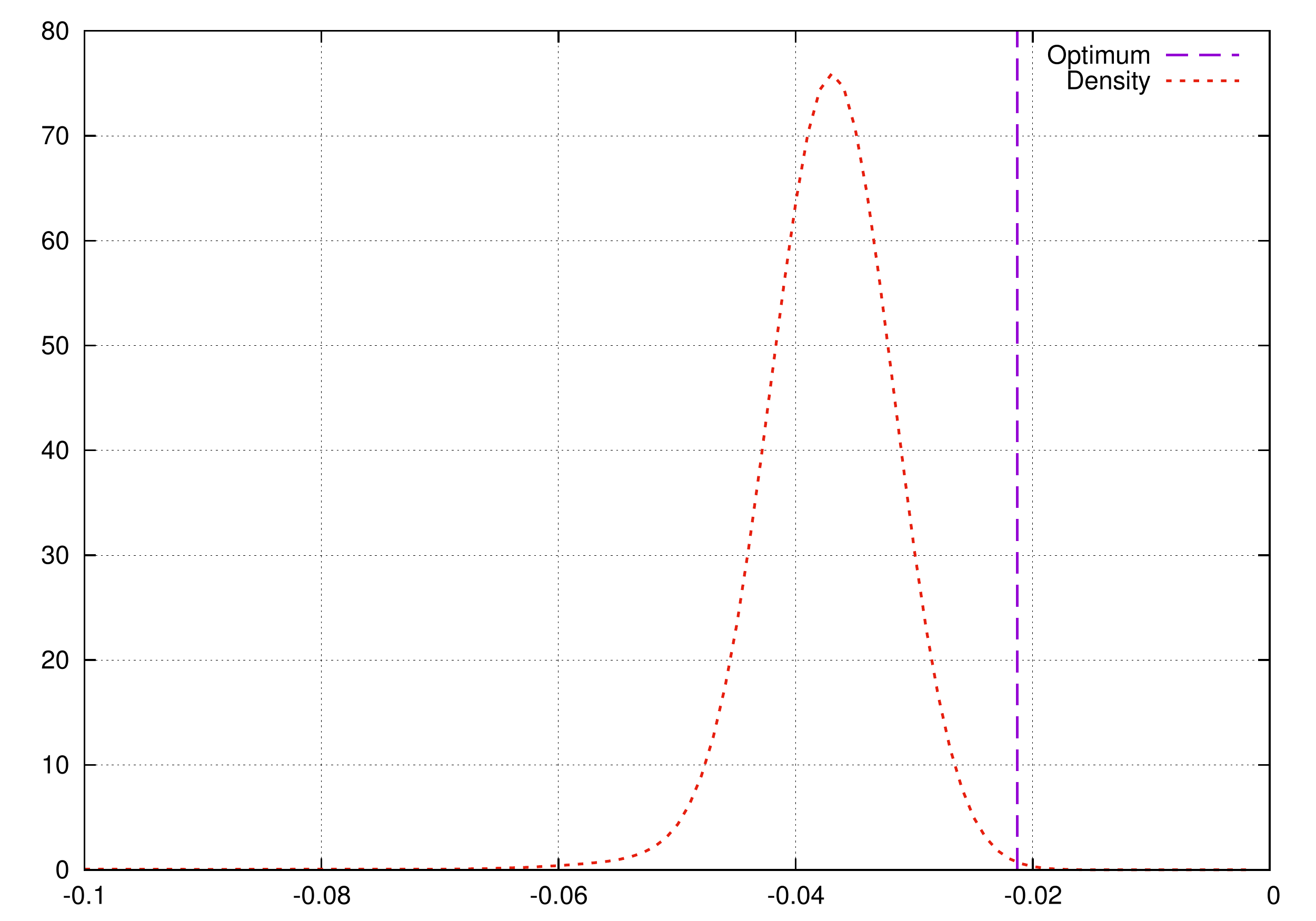}
		\end{center}
		\caption{
			Empirical density of \(-v(\balpha_\Gamma(\omega))\) (\(v_1 = v_2 = \beta =
			1\), \(\Gamma = 10000\), \(\lambda=0.04\)).
			The vertical line corresponds to the \(v(\balpha^*)\).}
		\label{figlesdensites}
	\end{figure}

	\subsection{Hedging error in the case \eqref{eq liquidity 2}}
	In this section, the liquidity cost function \(\Psi\) is chosen as in 
	\eqref{eq liquidity 2}.
	
	Figure~\ref{fig:compare_indicator_no_zoom} shows 
	\(-v(\balpha_\Gamma)\) for \(\Gamma=0\) (green) and 
	\(\Gamma=10E6\)
	(violet). 
	The initial parameter \(\balpha_0\) is the optimal one for a market without
	liquidity cost.
	
	We observe that the main
	part of the loss is saved thanks to the optimization procedure.
	In Figure~\ref{fig:compare_indicator_zoom}, we zoom on the 
	surface resulting from the optimization procedure.
	Notice that the cost function \eqref{eq liquidity 1} is equal 
		to the cost function \eqref{eq liquidity 2} in the particular case \(C=0\).
		So, Figure~\ref{fig:compare_indicator_zoom} allows to compare the value
		functions
		corresponding to these two cost functions.

	\begin{figure}[ht]
		\begin{center}
			\includegraphics[width=10cm,angle=-90]{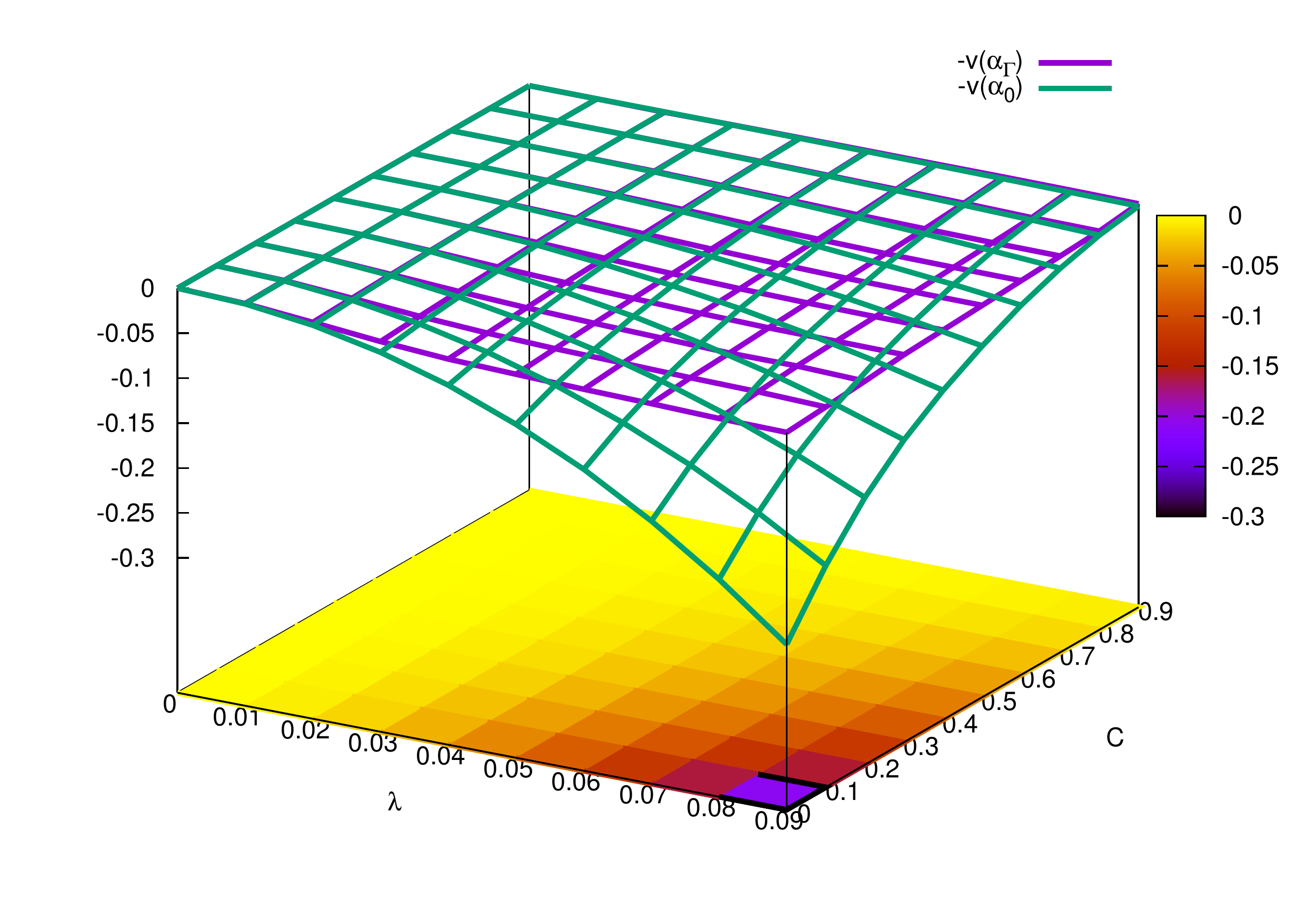}
			\caption{Violet surface: \(-v(\balpha_\Gamma)\) (\(\Gamma = 10E6\), \(v_1 = v_2 =
				100\), \(\beta = 0.6\)) for a liquidity cost function \eqref{eq liquidity 2} in
				terms of
				the value \(\lambda\) of the liquidity cost and the size \(C\) of the compact. 
				Green surface:\(-v(\balpha_0)\).}
			\label{fig:compare_indicator_no_zoom}
		\end{center}
	\end{figure}
	
	\begin{figure}[ht]
		\begin{center}
			\includegraphics[width=10cm,angle=-90]{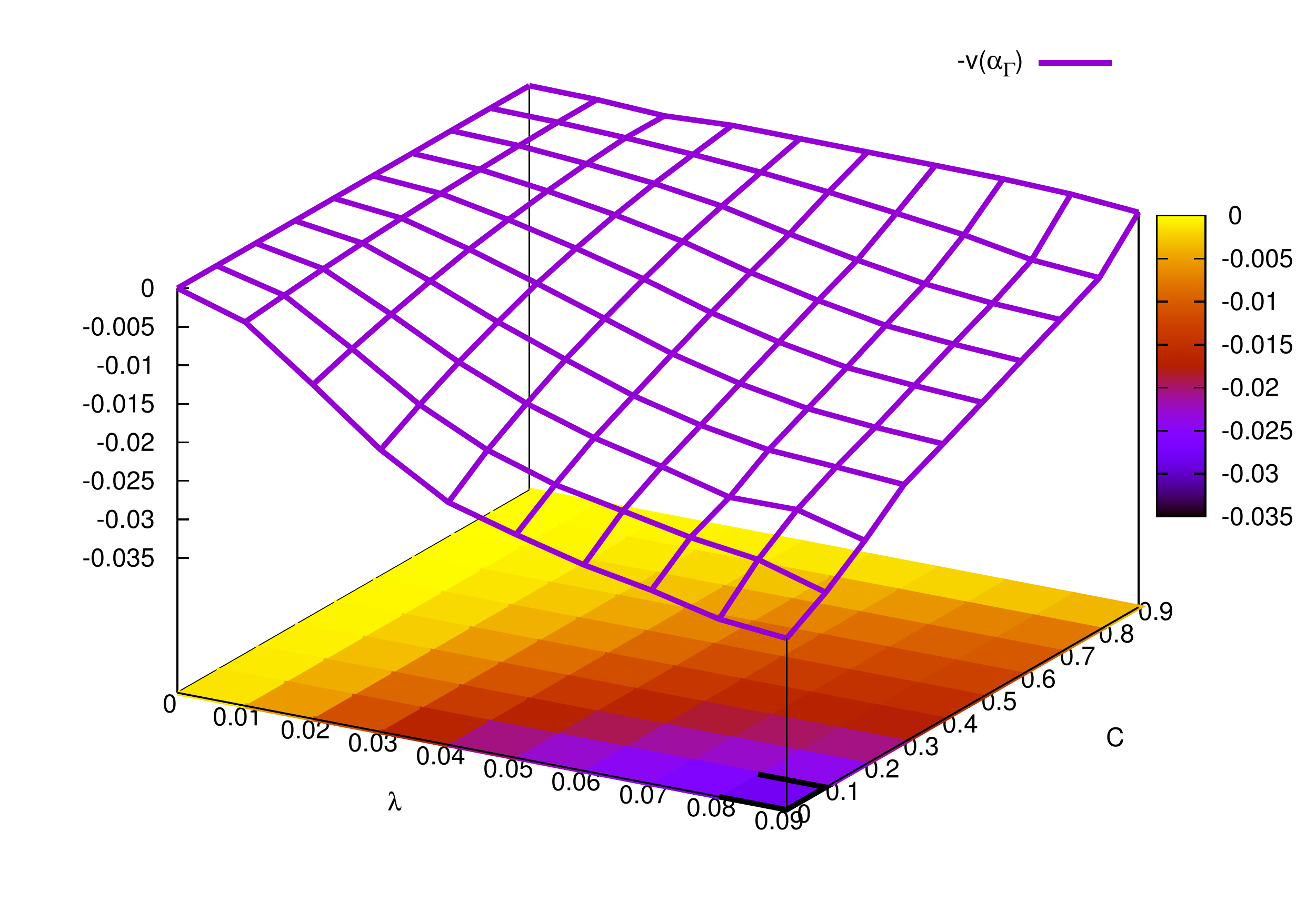}
			\caption{A zoom on the violet surface of
				Figure~\ref{fig:compare_indicator_no_zoom}}
			\label{fig:compare_indicator_zoom}
		\end{center}
	\end{figure}

	\subsection{Influence of the initial value \(\balpha_0\) of the optimization
		procedure}
	In Figure~\ref{fig:compare_different_strat}, we draw two functions 
	of the liquidity cost \(\lambda\): \(-v(\balpha_\Gamma)\)
	and \(-v(\boldsymbol{\delta}_\Gamma)\), where \((\balpha_\Gamma)\) and 
	\((\boldsymbol{\delta}_\Gamma)\) are the parameters obtained after \(\Gamma\) 
	steps of the optimization procedure but
	with different initial values \(\balpha_0\) and \(\boldsymbol{\delta}_0\) as
	described in 
	Subsec.~\ref{sub:50}. 
	The performance
	of the strategies obtained after \(\Gamma=10E6\) steps 
	are quite similar. It means that the sensitivity to the arbitrary initial 
	parameter \(\balpha_0\) is 
	%no more observable 
	not observable any more
	after \(\Gamma=10E6\) steps.
	\begin{figure}[ht]
		\begin{center}
			\includegraphics[angle=-90,width=0.9\textwidth]{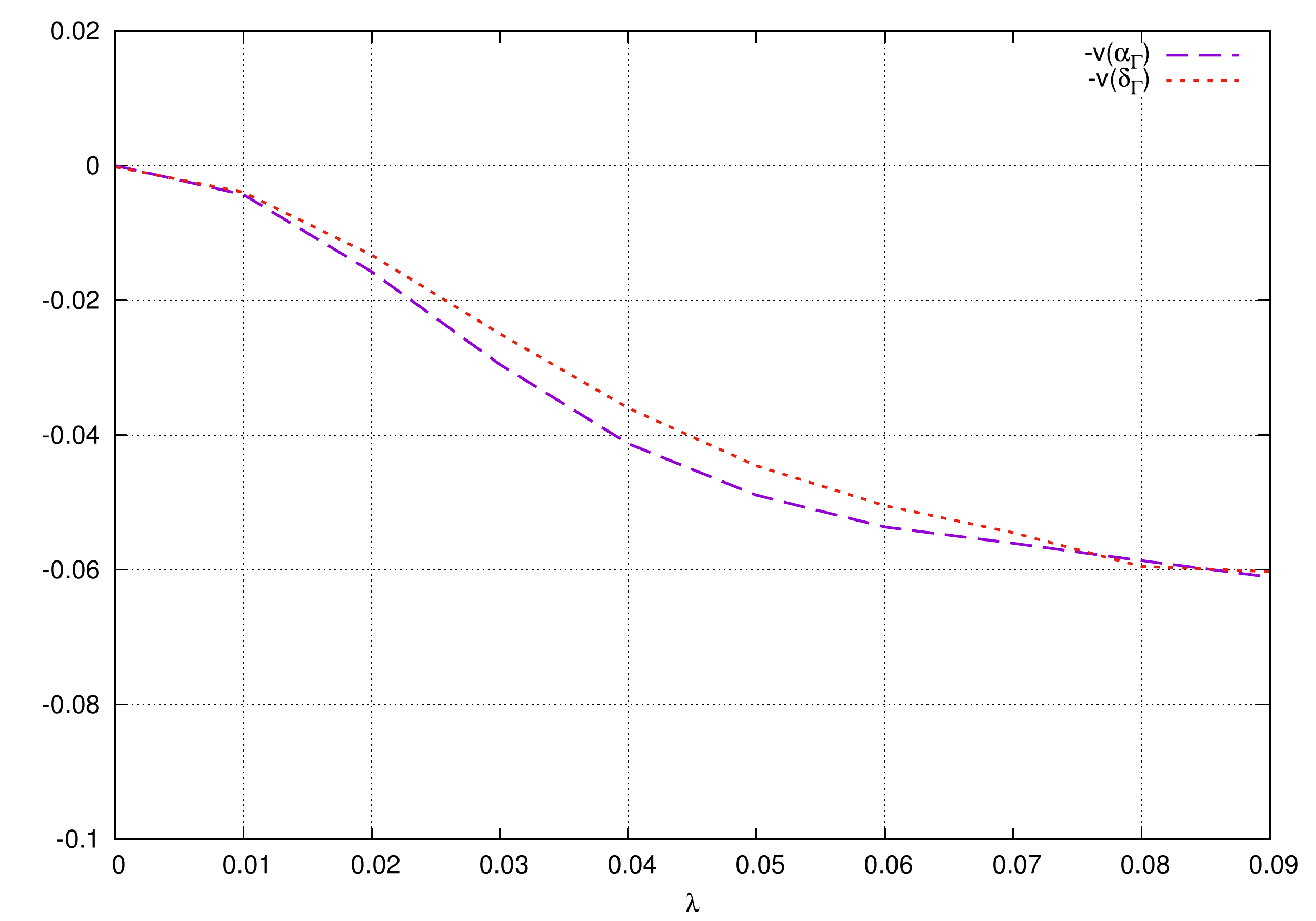}
			\caption{
				(short dashes) \(-v(\boldsymbol{\delta}_\Gamma)\)
				(long dashes) \(-v(\balpha_\Gamma)\), 
				in terms  of the liquidity cost \(\lambda\) 
				(\(\Gamma= 10E6\), \(v_1 = 0.1\), \(v_2 = 100\), \(\beta = 0.6\), \(d=3\)).
			}
			\label{fig:compare_different_strat}
		\end{center}
	\end{figure}
	\subsection{Reducing the set of admissible strategies}
	Recall Hypothesis~\ref{hypo:controlemesurable}.
	So far, our admissible strategies at time \(T_j\) 
	depend on all the past and present rates \(R_t,\cdots,R_{T_j}\).
	Thus the number of parameters
	\(\alphanjji\) to optimize is at least of the
	order of magnitude of the binomial coefficient
	\(\binom{N}{d}\), where \(N\) is the number of dates
	and the degree of truncation \(d\) 
	is defined as in \eqref{eq ensemble de troncation}.
	This order of magnitude is a drastically increasing function of
	\(N\).
	This crucial drawback leads us to try to simplify the complexity
	of the control problem \eqref{eq:optimizationgenerale} by reducing 
	the size of the set of the admissible strategies \(\Pi\).
	Observe that the optimal strategy under the perfect liquidity assumption
	has the property that \(\pi^*(j,i)\) only depends on \(R_{T_j}\).
	This observation suggests to face large numbers of dates
	by reducing the set of controls to controls depending only on a 
	small number of recent interest rates \(R_{T_j},
	R_{T_{j-1}},\cdots,R_{T_{j-q}}\).
	
	Figures~\ref{fig:FiveDatesFromZero} and \ref{fig:TenDatesFromZero} illustrate
	that the optimal control problem \eqref{eq:optimizationgenerale} with admissible
	
	strategies defined as in hypothesis~\ref{hypo:controlemesurable}
	may be used as benchmarks to solve control problems.
	Consider swaps with \(N=5\) and \(N=10\) dates of payment.
	In each one of these two cases, we study the effect of choosing
	\(q=0\) (that is at time \(T_j\), admissible strategies only depend on
	\(R_{T_j}\)), \(q=1\) (admissible strategies depend on \(R_{T_j}\) and
	\(R_{T_{j-1}}\)), \(q=2\). 
	
	Figure~\ref{fig:FiveDatesFromZero}  shows the performance of
	the corresponding   strategies obtained after \(\Gamma=10E6\)  steps of the
	optimization 
	algorithm for a swap with \(N=5\) dates of payments. 
	
	Figure~\ref{fig:TenDatesFromZero}
	shows similar quantities for a swap with \(N=10\)
	dates of payment. 
	We observe that the numerical computation of the optimal strategy is quite
	unstable
	when \(q\) is too big, which 
	reflects the difficulty to solve
	a high dimensional optimization problem. Therefore, one necessarily must choose 
	\(q\) small in order to get accurate approximations of optimal strategies
	belonging to 
	reduced sets of admissible strategies.

	\begin{figure}[ht]
		\begin{center}
			\includegraphics[width=10cm,angle=-90]
			{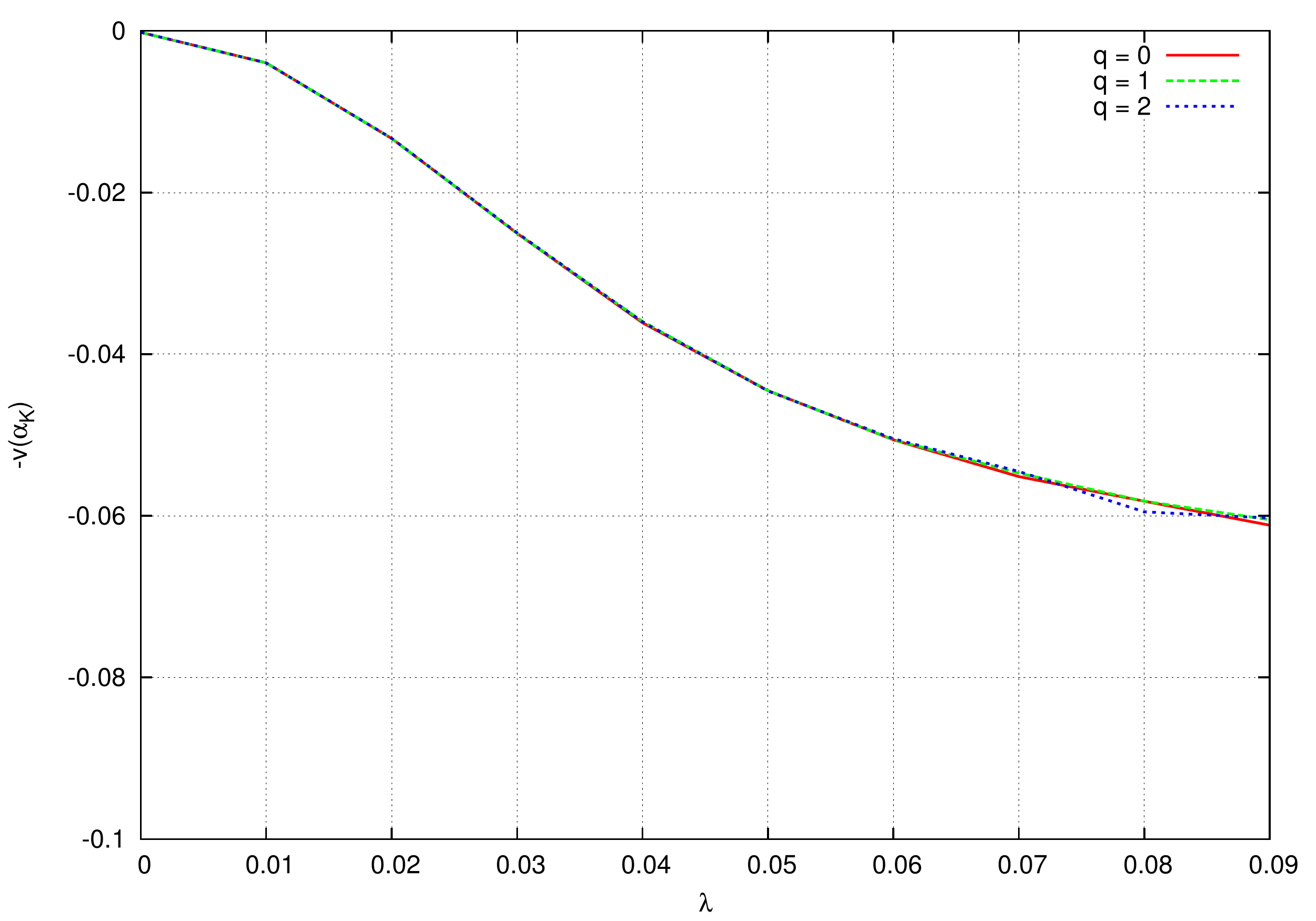}
			\caption{\(-v(\balpha_\Gamma)\) in terms of \(\lambda\)
				for strategies depending only on recent rates \(R_{T_j},\cdots,R_{T_{j-q}}\)
				(\(q=0,1,2\), 
				%\(N=5\), 
				\(\Gamma=10E6\), \(v_1 = 0.1\), \(v_2 = 100\), \(\beta = 0.6\), \(d=3\))
				for a swap with \(N=5\) dates of payment.
			}
			\label{fig:FiveDatesFromZero}
		\end{center}
	\end{figure}

	\begin{figure}[ht]
		\begin{center}
			\includegraphics[width=10cm,angle=-90]
			{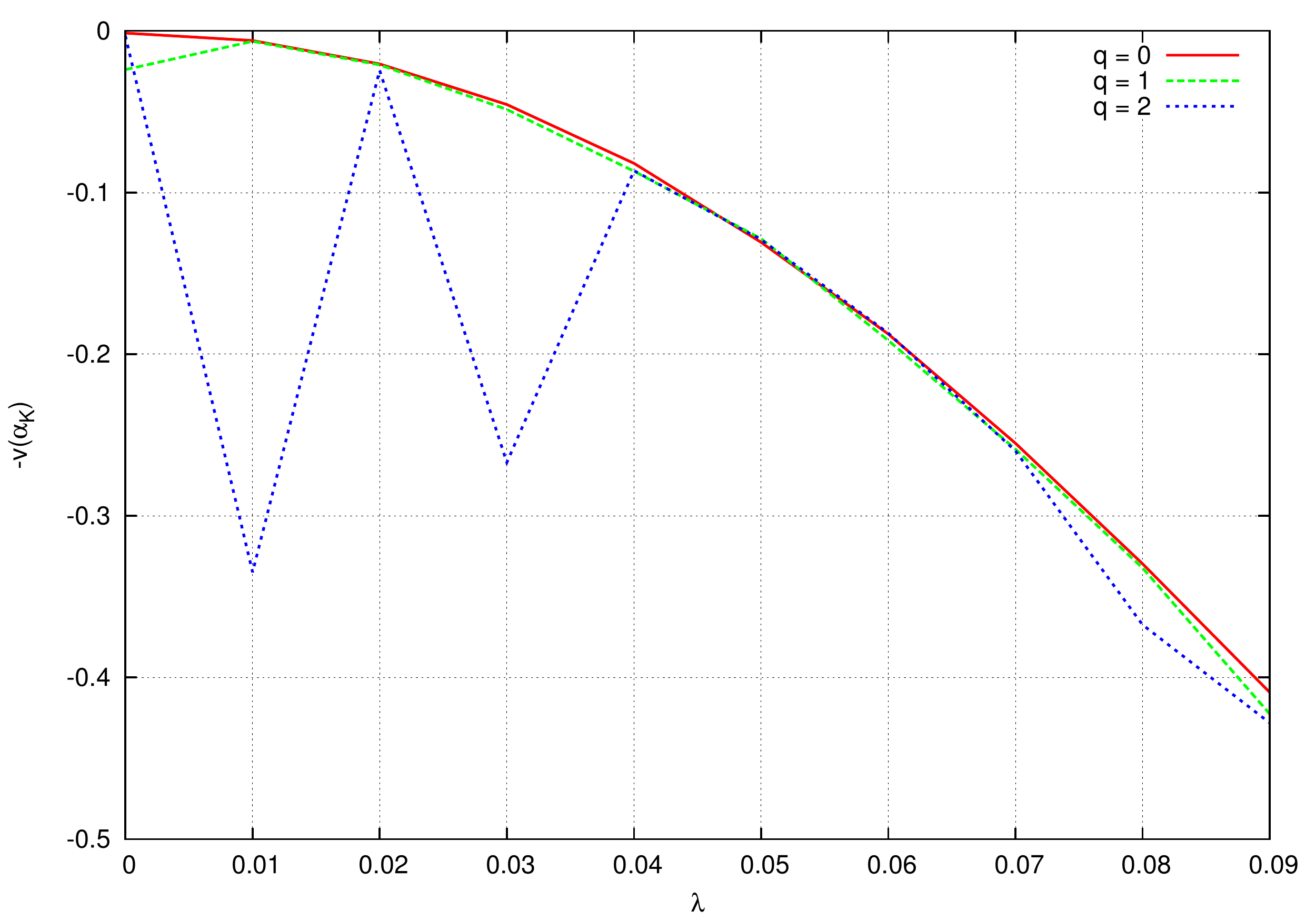}
			\caption{
				\(-v(\balpha_\Gamma)\)  in terms of \(\lambda\)
				for strategies depending only on recent rates \(R_{T_j},\cdots,R_{T_{j-q}}\)
				(\(q=0,1,2\), 
				%\(N=10\), 
				\(\Gamma=10E6\), \(v_1 = 0.1\), \(v_2 = 100\), \(\beta = 0.6\), \(d=3\))
				for a swap with \(N=10\) dates of payment.
			}
			\label{fig:TenDatesFromZero}
		\end{center}
	\end{figure}

	\clearpage

	\section{Conclusion}\label{sec conclusion}
	Stochastic control problems generally
	have no explicit solutions and are difficult to solve numerically.
	In this paper, we have proposed an efficient algorithm to 
	approximate optimal
	allocation strategies to hedge interest rate derivatives 
	%submitted
	subject
	to liquidity costs.
	
	As discussed above, our methodology is constructive and efficient
	in a Gaussian paradigm. We project the admissible allocation 
	strategies to the space generated by the first Hermite polynomials and
	use a classical stochastic algorithm 
	to optimally choose the coefficients of the projection in order to optimize
	an expected function of the terminal hedging error.
	
	We have illustrated this general approach by studying swaps 
	in the presence of liquidity costs.
	We have discussed  the performances of the numerical method in terms
	of all its algorithmic components.
	
	We emphasize that our methodology can be applied to many control problems,
	e.g., the computation of indifference prices, when the  model under
	consideration
	belongs to a Gaussian space.
	
		\section*{Acknowledgment}
		The authors would like to thank the anonymous referee for her/his careful
		reading 
		and useful remarks.

	\def\refname{References}

	\bibliographystyle{rAMF}
	\bibliography{biblio_calyon}
	
	\section{Appendix}\label{sec appendice}
	\subsection*{Proof of Lemma~\ref{lemma erreur troncature}}
	\begin{proof}
		Let us first prove the result for \(M=0\), that is \(
		X = \exp\left(\mu +  \lambda G\right).
		\)
		\begin{align*}
			X &= \sum_{n=0}^{+\infty}\alpha(n)\dfrac{H_n(G)}{\sqrt{n!}}
		\end{align*}
		where
		\begin{align*}
			\alpha(n) & = \mathbb{E}\left[X\dfrac{H_n(G)}{\sqrt{n!}}\right] = 
			\exp(\mu)\mathbb{E}\left[\exp(\lambda G)\dfrac{H_n(G)}{\sqrt{n!}}\right].
		\end{align*}
		We then use the identity:
		\( 
		e^{\lambda x - \lambda^2/2}
		=  
		\sum_{j \geqslant 0} \frac{\lambda^j}{j!} H_j \left(x\right)
		\) and obtain
		\begin{align*}
			\alpha(n) & = 
			\exp(\mu+\lambda^2/2)\mathbb{E}\left[\sum_j\dfrac{\lambda^j}{j!}H_j(G)\dfrac{H_n(G)}{\sqrt{n!}}\right].
		\end{align*}
		As \((H_j(G)/\sqrt{j!},j\geq 0)\) is an orthonormal basis of \(L^2(G)\) we have
		\[
		\alpha(n)  = 
		\exp(\mu+\lambda^2/2)\dfrac{\lambda^n}{\sqrt{n!}}.
		\]
		Let \(X^d\) be the projection of \(X\) on the subspace generated by the \(H_0,
		\cdots, H_d\)~:
		\begin{align*}
			X^{d} &= \sum_{n=0}^{d}\alpha(n)\dfrac{H_n(G)}{\sqrt{n!}}=
			e^{\mu +
				\lambda^2/2}\sum_{n=0}^{d}\dfrac{\lambda^n}{\sqrt{n!}}\dfrac{H_n(G)}{\sqrt{n!}}\\
			X - X^{d} &= e^{\mu +
				\lambda^2/2}\sum_{n=d+1}^{\infty}\dfrac{\lambda^n}{\sqrt{n!}}\dfrac{H_n(G)}{\sqrt{n!}}.
		\end{align*}
		The truncation error  is
		\begin{align*}
			\|X - X^{d}\|^2_2 &= e^{2\mu +
				\lambda^2}\sum_{n=d+1}^{\infty}\dfrac{\lambda^{2n}}{n!}\\
			&\leq e^{2\mu + \lambda^2}
			\dfrac{\lambda^{2(d+1)}}{(d+1)!}
			\sum_{n=0}^{\infty}\dfrac{\lambda^{2n}}{n!}\\
			&\leq e^{2\mu + 2\lambda^2}
			\dfrac{\lambda^{2(d+1)}}{(d+1)!}.
		\end{align*}
		The desired result thus holds true for \(M=0\). Let \(M\) be a positive
		integer. We have
		\[
		\|X - X^{d}  \|^2_2 = 
		e^{2\mu + 
			\lambda_0^2+\cdots+\lambda_M^2
		}
		\sum_{
			n_0+\cdots+n_M
			>d}\prod_{m=0}^{M}
		\dfrac{\lambda_m^{2n_m}}{(n_m)!}
		\]
		Recall the classical identity
		\begin{equation}\label{eq somme_produit}
			\dfrac{\left(a_0+\cdots+a_L\right)^N}{N!} = \sum_{n_0+\cdots+n_L =
				N}\prod_{m=0}^L\dfrac{a_m^{n_m}}{(n_m)!}.
		\end{equation}
		Thus,
		\begin{align*}
			\sum_{n_0+\cdots+n_M > d}\prod_{m=0}^{M}
			\dfrac{\lambda_m^{2n_m}}{(n_m)!}&=
			\sum_{N=d+1}^{\infty}\sum_{n_0+\cdots+n_{M}=N}\prod_{m=0}^{M}
			\dfrac{\lambda_m^{2n_m}}{(n_m)!}\\
			&=
			\sum_{N=d+1}^{\infty}\dfrac{\left(\lambda_0^2 + \cdots +
				\lambda_{M}^2\right)^N}{N!}\\
			&\leq \dfrac{\left(\lambda_0^2 + \cdots +
				\lambda_{M}^2\right)^{d+1}}{(d+1)!}\exp\left(
			\lambda_0^2 + \cdots + \lambda_{M}^2\right).
		\end{align*}
		This ends the proof for all positive \(M\).
	\end{proof}
\end{document}